\documentclass[acmsmall]{acmart}
\setcopyright{none}
\renewcommand\footnotetextcopyrightpermission[1]{}
\usepackage{adjustbox}
\usepackage{multirow}

\usepackage{amssymb}
\usepackage{amsmath}
\usepackage{amsfonts}
\usepackage{algorithm,algpseudocode}
\algrenewcommand\algorithmicindent{0.5em}
\usepackage{wrapfig,lipsum,booktabs}
\usepackage{subfigure}
\usepackage{float}  
\usepackage{graphicx}
\usepackage{textcomp}
\usepackage{tikz}
\usetikzlibrary{positioning,calc}
\usepackage{xcolor}
\usepackage{graphicx}
\usepackage{enumitem}
\usepackage{amsthm}
\usepackage[utf8]{inputenc}

\newtheorem{definition}{Definition}

\newcommand{\negl}{\operatorname{negl}}
\newtheorem{proposition}{Proposition}
\newtheorem{property}{Property}
\newtheorem{assumption}{Assumption}
\usepackage{hyperref}
\usepackage[compact]{titlesec}
\usepackage[mathlines,switch]{lineno}
\usepackage{tikz}
\usetikzlibrary{arrows.meta,positioning,shapes.geometric,calc}
\makeatletter
\algnewcommand{\LineComment}[1]{\Statex \hskip\ALG@thistlm \(\triangleright\) #1}
\makeatother

\newcommand{\cosmeticbf}{\textbf{CoSMeTIC}}
\newcommand{\cosmetic}{CoSMeTIC }

\begin{document}

\title{CoSMeTIC: Zero-Knowledge Computational Sparse Merkle Trees with Inclusion-Exclusion Proofs for Clinical Research}

\author{Mohammad Shahid}
\email{mohshai@okstate.edu}

\author{Paritosh Ramanan}
\authornote{Corresponding author. Under review at ACM Transactions on Privacy and Security}
\email{paritosh.ramanan@okstate.edu}

\affiliation{%
  \department{School of Industrial Engineering and Management}
  \institution{Oklahoma State University}
  \city{Stillwater}
  \state{Oklahoma}
  \country{USA}
}

\author{Mohammad Fili}
\email{mfili@gmu.edu}

\author{Guiping Hu}
\email{ghu4@gmu.edu}

\affiliation{%
  \institution{George Mason University}
  \department{Department of Systems Engineering and Operations Research}
  \city{Fairfax}
  \state{Virginia}
  \country{USA}
}

\author{Hillel Haim}
\email{hillel-haim@uiowa.edu}

\affiliation{%
  \institution{University of Iowa}
  \department{Department of Microbiology and Immunology}
  \city{Iowa City}
  \state{Iowa}
  \country{USA}
}

\begin{abstract}
Analysis of clinical data is a cornerstone of biomedical research with applications in areas such as genomic testing and response characterization of therapeutic drugs. Maintaining strict privacy controls is essential because such data typically contains personally identifiable health information of patients. At the same time, regulatory compliance often requires study managers to demonstrate the integrity and authenticity of participant data used in analyses. Balancing these competing requirements of privacy preservation and verifiable accountability remains a critical challenge. In this paper, we present CoSMeTIC, a zero-knowledge computational framework that proposes computational Sparse Merkle Trees (SMTs) as a means to generate verifiable inclusion and exclusion proofs for individual participants' data in clinical studies. We formally analyze the zero-knowledge properties of CoSMeTIC and evaluate its computational efficiency through extensive experiments.  We demonstrate the framework on Huntington's disease and HIV-1 case studies, using simulated CAG-repeat cohorts derived from published summary statistics and published de-identified clinical lab measurements of virus samples. Using two-sample Kolmogorov-Smirnov and likelihood-ratio hypothesis tests, along with logistic-regression-based genomic analyses on the de-identified datasets, we show that CoSMeTIC achieves strong privacy guarantees while maintaining statistical fidelity.   Our results suggest that CoSMeTIC provides a scalable and practical alternative for achieving regulatory compliance with rigorous privacy protection in large-scale clinical research.
\end{abstract} 

\maketitle

%



\section{Introduction}
Clinical research studies rely on a wide variety of statistical tools and techniques to analyze patient outcomes so as to address the underlying biomedical research questions. Critically, these research questions can potentially help in deciphering the propensity of people to develop certain diseases \cite{MengZhang2024,Qiu2025}; formulate statistically robust strategies for their disease prevention \cite{Guerra2025,Bi2022}; as well as characterize or benchmark responses to therapeutic drugs on a wide demographic base \cite{Huang2018,Xie2017}. From a computational perspective however, the underlying techniques could potentially involve using statistical tests to help guide the core clinical research hypothesis \cite{Shih2017,Rojas2024}; or utilize comparative benchmarking of different statistical models on varying demographic subgroups \cite{Ward2020}. Therefore, computational correctness and data accountability are the two fundamental aspects that govern the integrity of the underlying statistical technique in clinical research studies \cite{Billah2025,Kelly2019}. Transparently asserting computational correctness and data accountability involves patient privacy risks \cite{Curzon2021,Malin2010}, significant audit burdens on regulating authorities \cite{Falco2021} as well as an inability to reliably claim patient exclusion \cite{liu2021evaluating,food2024evaluating}. In this paper, we develop \cosmeticbf, a novel framework to assert computational correctness and user data privacy and accountability in clinical research studies using zero knowledge driven computational sparse Merkle trees (CSMTs).

The \cosmetic framework is primarily designed to resolve the three important critical gaps that exist in state of the art clinical research studies. First, ensuring the computational correctness of the study outcome becomes an essential precursor to guarantee the integrity of the clinical study. We define correctness in terms of the order as well as the precision of each computational step that belies the underlying statistical method. Asserting computational correctness in a transparent fashion can potentially help reduce the audit burden on regulatory bodies (such as the U.S. FDA), especially in large scale clinical studies that are geographically spread out. Second, the clinical study must also guarantee data accountability by demonstrating that the study outcomes are based on genuine patient data records. As part of the data accountability argument, the clinical study manager must also be capable of providing individual patients an inclusion or exclusion receipt pertaining to the use of their respective data records. Lastly, the entire computational graph from the raw user data to the final statistical outcome must be publicly verifiable in a transparent fashion, without the need to divulge identities of the patients or the users. 





The \cosmetic framework is specifically designed to address problems of computational correctness, patient data privacy and accountability as well as transparency for regulatory authorities. For ensuring correctness, we first decompose the statistical technique into individual computational steps. Consequently, the \cosmetic framework leverages computational reductions (CR) which are regarded as the basic building blocks for statistics and data analytics \cite{dean2008mapreduce,dongarra1995introduction}. CR primitives have been widely used for carrying out large-scale data science and analytics pipelines and are the core drivers behind large scale compute frameworks such as Apache Spark \cite{zaharia2010spark} as well as Message Passing Interface (MPI) \cite{dongarra1995introduction}. As a result, a series of CR primitives can be used to formally map each computational step involved in a statistical technique recursively from raw user data to the final statistical output. In this paper, we leverage CSMTs to represent each CR operation involved in the statistical method to enable membership proofs for distinct computational steps.


Conventionally, plain sparse Merkle Trees (SMTs) are defined on the basis of leaf space cardinality which must be equivalent to the output bit space of the hashing algorithm employed to construct the tree. By notionally equating the set of potential leaves with the set of total possible outcomes of a hash algorithm, an SMT can deliver inclusion as well as exclusion proofs for individual data records \cite{tzialla2021transparency,kurbatov2024imok}. However, the conventional SMT architecture relies on simple hash concatenation applied at each level of the tree which is fundamentally incompatible with the CR operations that seek to aggregate data across the entire clinical user dataset. As a result, it becomes significantly challenging to provide a verifiable computational trace regarding the statistical methodology of clinical studies while retaining user privacy and clinical audit efficiency.

In order to alleviate the CR-oriented limitations of the Merkle tree architecture, we propose the CSMT architecture that promises the potential of delivering user membership proofs with respect to individual computational steps themselves. In \cosmetic each CSMT represents one CR step with the corresponding primitive applied at each level of the Merkle tree. As a result, the \cosmetic framework enables user inclusion and exclusion proofs at each CR step of the statistical technique used in the clinical study. To induce transparency, we augment the CSMT architecture with publicly verifiable, zero-knowledge proofs that can assert the integrity of the data transformation and aggregation with complete privacy of the individual user. These zero-knowledge proofs rely on the zk-SNARK (Succinct, Non-interactive, Argument of Knowledge) paradigm and help assert the overall integrity of the entire statistical technique, capturing every transformation of raw user data records among a series of zk-SNARKs.

The \cosmetic framework delivers key benefits for both regulators and patients by addressing three critical gaps relating to computational correctness, data accountability, and privacy-preserving transparency. For regulators, \cosmetic enables efficient public verification of the correct execution of statistical methods underlying clinical study claims, significantly reducing audit burdens and potentially accelerating regulatory clearance. For patients and data contributors, the framework provides strong data governance guarantees through verifiable inclusion receipts and exclusion proofs, ensuring consent and accountability. Finally, publicly verifiable zk-SNARK proofs for leaf transformations and Merkle tree traversal offer transparent, end-to-end visibility into the study’s computational pipeline, allowing patients to understand precisely how their data was used and enabling regulators to pinpoint the exact role of individual records within the overall analysis.

Our work showcases the efficacy of the \cosmetic framework with respect to different statistical methods as well as real world case studies.  We demonstrate the framework on illustrative datasets, Huntington's disease CAG-repeat distributions and HIV-1 Temsavir resistance measurements, chosen to represent common genomic and virus resistance analysis workflows.   Using these datasets we formulate clinical research studies involving the Kolmogorov-Smirnov (KS) statistical hypothesis tests as well as accuracy and Likelihood Ratio tests (LRT) to benchmark the predictive performance and fitting quality of logistic regression models, respectively. Our contributions can be summarized as follows:
\begin{itemize}[topsep=2pt, itemsep=2pt, parsep=0pt]
    \item We develop a Computational Sparse Merkle Tree (CSMT) module that integrates CR primitives with the Merkle tree architecture that is capable of providing inclusion and exclusion proofs for individual computational steps as a means for privacy and accountability.
    \item We build a zk-SNARK driven framework that encodes the entire graph of CR primitives in zero-knowledge providing end-to-end transparency of the clinical study.
    \item We provide theoretical security guarantees regarding the robustness of the zero-knowledge framework that delivers patient data privacy, accountability and transparency.
    \item We demonstrate the performance of the framework on a wide range of statistical methods involving the KS hypothesis test as well as the LRT and accuracy tests for logistic regression applied towards the clinical datasets.
\end{itemize}

We now proceed to summarize the related work pertaining to SMTs in addition to a brief review of current state of the art methods in the domain of verifiable statistics.

\section{Related Work}
In recent years, there are strong emerging research trends pertaining to verifiable statistics. Custodes \cite{servan2018cryptographically} certifies classical hypothesis tests (e.g., $t$-tests, $\chi^2$ tests, ANOVA) by logging encrypted test evaluations and issuing cryptographically certified $p$-values, thereby deterring $p$-hacking. Verifiable differential privacy (VerDP) \cite{narayan2015verdp} and its follow-ups use zero-knowledge proofs to show that randomized mechanisms (e.g., Laplace noise for counts) have been applied correctly to sensitive data before releasing aggregate statistics.  RiseFL~\cite{zhu2024risefl} embeds a $\chi^2$-style goodness-of-fit test into a low-cost SNARK circuit to verify that federated-learning updates satisfy certain norm bounds, and ElectionGuard-based protocols~\cite{lonfils2023electionguard} combine per-ballot zero-knowledge proofs with risk-limiting audits to statistically validate election outcomes. 

While all of these works focus on verifiable statistics, they lack the ability to account for datum-level accountability. These systems demonstrate that hypothesis tests and related statistical procedures can be implemented inside zero-knowledge circuits, but they treat the underlying dataset as a monolithic object. The cryptographic guarantees apply to the correctness of the \emph{global} test statistic or mechanism, not to verifiable claims about whether any specific individual’s datum was included or excluded from the analysis. 

Existing work that targets privacy-preserving user-level data accountability, has made significant progress but only in terms of user-data oriented state accountability. Transparency dictionaries such as Verdict \cite{tzialla2021transparency} introduce indexed Merkle trees that support succinct inclusion and non-inclusion proofs over a $2^{256}$-sized key space and can be combined with SNARK backends for lazy, on-demand verification.  IMOK \cite{kurbatov2024imok} extends sparse Merkle trees (SMTs) with non-prohibition proofs for sanction lists, while Cartesian Merkle Trees (CMTs) \cite{chystiakov2025cmt} and related constructions optimize membership and non-membership proofs for blockchain and rollup settings. Recent compliance-oriented systems, such as private smart wallets with probabilistic compliance~\cite{rizzini2025private}, likewise use SMTs to certify that a given account is (or is not) contained in a prohibited set. 

However, in all of these designs, Merkle leaves are treated as \emph{static key–value records}, and the associated zero-knowledge circuits implement only lookup, append-only, or simple boolean predicates. None of these systems implements a \emph{computational} SMT or Merkle-sum SMT that is capable of supporting arbitrary per-leaf transformations and numeric reductions over selected leaves. Further, they lack the capability of providing inclusion and exclusion proofs with respect to the aggregated reductions over individual datum belonging to specific users.  Proof-of-solvency protocols based on Merkle-sum trees~\cite{snarkedsummtree} do realize aggregation semantics at internal nodes, but typically over dense trees, without published non-membership circuits, and in financial rather than statistical-testing settings. Therefore, there exists a critical gap that needs to be bridged between the \emph{state accountability} of the data, and the \emph{analysis accountability} at the granularity of individual datum that is required in clinical studies. Additionally, for implementing these missing capabilities, we need novel methodologies that scale to sparse $2^{256}$ key spaces while keeping proof generation lazy and participant-driven.


Taken together, the two strands of verifiable statistic and analysis accountability leave a critical gap for privacy-sensitive scientific domains such as clinical studies. On one side, SMT-based ZK frameworks provide efficient inclusion and non-inclusion proofs but are largely agnostic to the statistical computations performed over the committed dataset, and they lack Merkle-sum semantics tailored to complex tests.  On the other side, ZK-based verifiable statistics frameworks ensure that certain hypothesis tests are executed correctly, but they do not cryptographically bind the test inputs to a fine-grained, participant-auditable commitment structure. In particular, they do not offer a way for a participant to ask, \emph{ex post}, “Was my data used in this KS test, likelihood-ratio test, or regression model?” and receive a provable inclusion or non-inclusion proof tied to that exact computation. As a result, there is a clear lack of supporting methodologies that can provide individual-level data accountability with provable membership guarantees in aggregations over pre-defined user datasets. 

  Specifically in the domain of membership, there are two particular approaches, ELEKTRA \cite{len2023elektra} and PARAKEET \cite{malvai2023parakeet} that provide the capability to authenticate with respect to public-key directories using variants of Merkle trees. However, these approaches are applicable to end-to-end encrypted communication paradigms. Consequently, these frameworks serve as authentication systems which cannot be used for demonstrating membership semantics across computational aggregation steps that are commonly found across a wide range of statistical approaches. Similarly, BalanceProofs \cite{wang2023balanceproofs} aggregates \emph{proofs} over a vector commitment to enable fast batch verification in proof-of-solvency settings; it operates over dense trees, offers no non-membership proofs, and aggregates committed proof artifacts rather than the underlying data values themselves. However, the BalanceProof approach does not certify exclusion from specific statistical pipelines which would explicitly require support for aggregation semantics over arbitrary operations. 

To the best of our knowledge, there is no existing framework that (i) instantiates a \emph{computational} sparse Merkle tree in which internal nodes encode analysis-specific reductions (e.g., transformed sufficient statistics for KS or likelihood-ratio tests), (ii) supports both inclusion and exclusion proofs for individual participants with respect to a given statistical pipeline, (iii) provides a scalable extensible framework for applications with a large number of aggregation operations, and (iv) demonstrates these guarantees on real-world clinical datasets without degrading the fidelity of downstream hypothesis tests. CoSMeTIC is designed to precisely fill this gap by combining SMT-based, participant-centric accountability with zero-knowledge implementations of classical hypothesis tests and regression models used in large-scale clinical research.



\section{Problem Formulation}
We consider a group of $n$ users denoted by $U = \{u_1, u_2, \ldots, u_n\}$ wherein each user $u_i$ possesses the raw datum $\delta_i \in \mathbb{R}^d$ leading to a user group dataset denoted by $\Delta_U$ as denoted in Equation \eqref{eq:raw_ug}.
\begin{gather}\label{eq:raw_ug}
    \Delta_U = \Big[\delta_1\ldots \delta_n\Big]
\end{gather}
Further, we consider a computational scheme consisting of a leaf transformation function denoted by $\mathcal{L}$ and an aggregator function denoted by $\mathcal{A}$ parametrized by parametrized by $\theta_{\mathcal{L}},\theta_{\mathcal{A}}$ respectively as defined in Equations \eqref{eq:r_ltr}, \eqref{eq:r_agg}. Their composite can be denoted as $\mathcal{A} \circ \mathcal{L}$ as defined in Equation \eqref{eq:r_ltr_agg_combined} respectively.
\begin{gather}
    \mathcal{L}:\mathbb{R}^{d}\mapsto \mathbb{R}^p \label{eq:r_ltr}\\
    \mathcal{A}:\mathbb{R}^{p\times n} \mapsto \mathbb{R}^p \label{eq:r_agg}\\
    \mathcal{A} \circ \mathcal{L}: \mathbb{R}^{d\times n} \mapsto \mathbb{R}^p \label{eq:r_ltr_agg_combined}
\end{gather}
The leaf transformation $\mathcal{L}$ is geared towards transformation of raw datum of individual users leading to a user group specific transformed dataset $\Delta^{\mathcal{L}}_U$ as denoted in Equation \eqref{eq:ltr_set}. On the other hand, the application of the aggregation $\mathcal{A}$ on set $\Delta^{\mathcal{L}}_U$ results in a reduction denoted by $R(\mathcal{A},\mathcal{L},U)$ as given in Equation \eqref{eq:global_agg_ltr}. 
\begin{gather}
    \Delta^{\mathcal{L}}_U = \big[\mathcal{L}(\delta_1;\theta_{\mathcal{L}}),\mathcal{L}(\delta_2;\theta_{\mathcal{L}}), \ldots \mathcal{L}(\delta_n;\theta_{\mathcal{L}})\big] \label{eq:ltr_set}\\
    R(\mathcal{A},\mathcal{L},U) = \mathcal{A}\Big(\Delta^{\mathcal{L}}_U;\theta_{\mathcal{A}}\Big) \label{eq:global_agg_ltr}
\end{gather}
The reduction $R(\mathcal{A},\mathcal{L},U)$ represents the aggregated outcome (e.g., statistical summary, model update, etc.) computed from the transformed dataset.

In order to formalize the objective of clinical stakeholders, we first define computational membership using function $\mathcal{M}$ as presented in Definition \ref{def_mf}. 
\begin{definition}[Computational Membership]\label{def_mf}
Given an arbitrary user identity $\tilde{u}$ with raw datum $\tilde{\delta}$, the membership function can be defined by $\mathcal{M}$ 
\[
\mathcal{M}([\tilde{u},\tilde{\delta}]|R) = \begin{cases}
        1 , \ \text{if } \ \mathcal{L}(\tilde{\delta}) \in \Delta^{\mathcal{L}}_U \text{ and } R = \mathcal{A}\Big(\Delta^{\mathcal{L}}_U\Big)\\
        0, \ \text{otherwise}
    \end{cases}\label{eq:mem}
\]
where  $U$, $\Delta_U$ denote a pre-defined user identity set and its corresponding user group dataset for a given leaf transform $\mathcal{L}$ and an aggregation function $\mathcal{A}$
\end{definition}
With the help of Definition \ref{def_mf}, we formally state the prover objective using Definition \ref{def_pf} for an arbitrary user identity $\tilde{u}$ with a datum corresponding to $\tilde{\delta}$.
\begin{definition}[Clinical Stakeholder Objective]\label{def_pf}
Given a user identity set $U$, its corresponding user group dataset $\Delta_U$ and the resulting global reduction $R(\mathcal{A},\mathcal{L},U)$ the clinical stakeholder objective is to provide a publicly verifiable proof of membership or non-membership of an arbitrary user identity $\tilde{u}$ by evaluating the membership function $\mathcal{M}([\tilde{u},\tilde{\delta}]|R(\mathcal{A},\mathcal{L},U))$ without disclosing $\Delta_U$ or any other intermediate transformations.
\end{definition}

Definition \ref{def_pf} provides a fundamental overview of the task of the clinical stakeholder which is responsible of handling patient datasets. Regulatory mandates would encumber the clinical stakeholder to demonstrate the use of data pertaining to specific individuals. For the regulator, this would help enforce compliance requirements regarding the integrity of the underlying datasets used in clinical trials.


\subsection{User Identity Management}
For implementing the SMT, we must first ensure that raw user datum is uniquely coupled with their corresponding identity. Additionally, depending on the nature of $\mathcal{L}$, there is a non-trivial probability of the user salted data of two users mapping to the same transformed output. Therefore, we must also ensure that transformed outputs are uniquely distinguishable across different users.

\noindent\textbf{Asserting User Identity}: We rely on user salts to uniquely assert the identity of the user as well as to enable them to demonstrate ownership of their datum as elaborated in Property \ref{prop:user_identity}. 
\begin{property}\label{prop:user_identity}
The identity of a user $u_i\in U$ can be uniquely bound to their corresponding datum $\delta_i$ through a secret salt vector $\mu_i\in \mathbb{R}^{s_u}$. 
\end{property}
\noindent Specifically, for each user identity $u$, the corresponding datum is concatenated with a user salt $\mu\in\mathbb{R}^{s_u}$ leading to $\delta^s \equiv (\delta,\mu)$. The user salt $\mu$ is a secret vector of  $s_u$ dimension uniquely held only by the user and must not be revealed or made public. As a consequence we obtain a salted user group dataset denoted by $\Delta^{s}_U$ in Equation \eqref{eq:salted_ug}. 
\begin{gather}
    \Delta^{s}_U = \Big[(\delta_1,\mu_1),(\delta_2,\mu_2)\ldots (\delta_n,\mu_n)\Big]\label{eq:salted_ug}\\
    \hat{\mathcal{L}}:\mathbb{R}^{{d+s_u}}\mapsto\mathbb{R}^{p} \label{eq:i_mod_ltr}
\end{gather}
The unique salt binding for each user ensures that each record in $\Delta^{s}_U$ is uniquely attributable to its originating user. In order to handle the transformations of the salted user datum, we modify the leaf transform input dimension to yield $\hat{\mathcal{L}}$ parametrized by $\theta_{\hat{\mathcal{L}}}$ as denoted in Equation \eqref{eq:i_mod_ltr}.

\noindent\textbf{Asserting Uniqueness of Leaf Transformation}: For enforcing distinguishability between transformed outputs of different users, we concatenate a user specific transform salt to the leaf transform output of salted user datum as described by Property \ref{prop:leaf_identity}. 
\begin{property}\label{prop:leaf_uniquenes}
The leaf transformation of the salted user datum $(\delta_i,\mu_i)$ of user $u_i\in U$ can be made uniquely distinguishable across user set $U$ by binding the outputs of the modified leaf transform function $\hat{\mathcal{L}}(\delta_i,\mu_i)$ with a secret transform salt vector $\tau_i\in \mathbb{R}^{s_t}$.
\end{property}
\noindent Specifically, for a given user $u$ with raw datum and user salt represented by $\delta,\mu$ respectively, we incorporate $\tau\in \mathbb{R}^{s_t}$ such that Equation \eqref{eq:salted_func_ltr} holds.
\begin{gather}
    \mathcal{L}^{s}\big([\delta,\mu,\tau];\theta_{\mathcal{L}^s}\big) \;=\; [\,\hat{\mathcal{L}}(\delta,\mu;\theta_{\hat{\mathcal{L}}}), \tau\,], \quad \tau \in \mathbb{R}^{s_t} \label{eq:salted_func_ltr}
\end{gather}

Consequently, we obtain a user group specific salted leaf transform dataset pertaining to set of users $U$ as denoted by $\Delta^{\mathcal{L}^s}_U$ in Equation \eqref{eq:salted_ltr_final}.
\begin{gather}
        \Delta^{\mathcal{L}^s}_U = \big[\mathcal{L}^s([\delta_1,\mu_1,\tau_1];\theta_{\mathcal{L}^s}), \ldots \mathcal{L}^s([\delta_n,\mu_n,\tau_n];\theta_{\mathcal{L}^s})\big] \label{eq:salted_ltr_final}
\end{gather}

As a result of using the transform salt, the output dimension of the leaf transform will also change effectively changing the dimensions of the aggregator as well. Therefore, the modified leaf transform and aggregator function can be denoted would acquire the mappings represented in Equations \eqref{eq:salted_ltr}-\eqref{eq:salted_ltr_agg_combined}.
\begin{gather}
    \mathcal{L}^{s}:\mathbb{R}^{d+s_u}\mapsto \mathbb{R}^{p+s_t} \label{eq:salted_ltr}\\
    \mathcal{A}^{s}:\mathbb{R}^{(p+s_u)\times n} \mapsto \mathbb{R}^{p+s_t} \label{eq:salted_agg}\\
    \mathcal{A}^{s} \circ \mathcal{L}^{s}: \mathbb{R}^{(d+s_u)\times n} \mapsto \mathbb{R}^{p+s_t} \label{eq:salted_ltr_agg_combined}
\end{gather}

\subsection{Computational Sparse Merkle Trees}
We consider Sparse Merkle Trees (SMT) as a means to represent the computational graph pertaining to the global aggregation function $\mathcal{A}^s$ parametrized by $\theta_{\mathcal{A}^s}$. SMTs are a sparse version of conventional Merkle trees which rely on a well-formed hash function formally as defined in Definition \ref{defn2}.
\begin{definition}[Hash Function]\label{defn2}
A hash function denoted by \verb|hash| can be defined such that it can consume arbitrary length inputs so as to map to fixed length outputs of size $K$
\[
\verb|hash| : \{0,1\}^* \rightarrow \{0,1\}^K
\]
\end{definition}
 Specifically in the context of Definition \ref{defn2}, we assume \verb|hash| is a cryptographic hash function satisfying collision and preimage resistance. We further assume the user and transform salts $\mu_u$ and $\tau_u$ carry sufficient entropy that the space of possible inputs to $\mathsf{hash}(\mathcal{L}^s(\delta_u, \mu_u, \tau_u))$ is too large for an adversary to enumerate by brute force, even if the underlying participant attributes $\delta_u$ are drawn from a small or otherwise guessable set of values. Together, these two properties ensure that an adversary observing $\mathsf{hash}(\mathcal{L}^s(\delta_u, \mu_u, \tau_u))$ can neither invert the hash directly, by preimage resistance, nor recover $\delta_u$ through exhaustive search over candidate inputs, by the entropy contributed by the salts. Consequently, we assume that the hashed leaf transforms and aggregation inputs do not reveal the underlying participant data to a computationally bounded adversary. 



Sparse Merkle Trees (SMTs) differ from conventional Merkle trees in both tree height and the number of leaf nodes. For a hash function producing an output of $K$ bits, a traditional Merkle tree has a variable height that depends on the number of data elements included as leaves. In contrast, an SMT consists of $2^K$
leaves resulting in a fixed tree height of $K$. Given a set of data elements, the construction of an SMT relies on determining the hash string of each element. The binary representation of each hash string determines the leaf position of the corresponding data element while the unfilled leaf positions are characterized by the hash of a default element. 

 Starting from the leaves, the conventional version of the SMT is constructed by recursively hashing the concatenation of hash strings of two adjacent elements ultimately culminating in the root element. Unlike conventional Sparse Merkle Trees (SMTs), which recursively concatenate and hash child nodes, the \emph{computational} variant performs a reduction operation at each recursion level. Specifically, we view the global aggregation function $\mathcal{A}^s$ as a recursive composition of local reduction functions $\mathcal{A}^{l}(\cdot;\theta_{\mathcal{A}^l})$, where each $\mathcal{A}^{l}$ combines the aggregated values of two child nodes into a higher-level representation as denoted in Equation \eqref{eq:a_rec}.  


\begin{equation}\label{eq:a_rec}
    \begin{aligned}
            \mathcal{A}^s&(\Delta^{\mathcal{L}^s}_U;\theta_{\mathcal{A}^s}) = \mathcal{A}^{l}_{(K)}\!\big(
  \mathcal{A}^{l}_{(K-1)}\!\big(
    \cdots 
    \mathcal{A}^{l}_{(0)}(\varphi^{L}_L, \varphi^{L}_R; \theta_{\mathcal{A}_l})
    \cdots;\theta_{\mathcal{A}_l}
  \big);\theta_{\mathcal{A}_l}
\big)
    \end{aligned}
\end{equation}

This recursive construction yields a \emph{Computational Sparse Merkle Tree (CSMT)}, formally defined in Definition~\ref{def:csmt}, for a given user identity set $U$ and their corresponding salted leaf transform set $\Delta^{\mathcal{L}^s}_U$. The relations contained in Definition \ref{def:csmt} is illustrated in Figure \ref{fig:csmt_leaf} using three leaf nodes.

\begin{definition}[Computational Sparse Merkle Tree (CSMT)]\label{def:csmt}
Given an aggregation function $\mathcal{A}^{l}$ parametrized by $\theta$, a CSMT is defined as an augmented Sparse Merkle Tree with the following properties.
\[
\begin{aligned}
    H_i &= \verb|hash|\big(\mathcal{L}^{s} (\delta_i,\mu_i,\tau_i)\big)\qquad \forall u_i\in U\\
    \varphi^0_j &= 
    \begin{cases}
        \mathcal{L}^{s}(\delta_i,\mu_i,\tau_i), \text{ if } \ \exists u_i\in U \text{ and }  H_i = \verb|Bin|(j)\\
        \mathcal{L}^s(\varnothing), \text{ otherwise }
    \end{cases}\\
    \varphi^k_P &= \mathcal{A}^{l}(\varphi^k_L, \varphi^k_R; \theta)\qquad\forall k\in \{0,\ldots,2^{K}-2\}\\
    H^k_P &= \verb|hash|\big(\varphi^k_P\big) \qquad\forall k\in \{0,\ldots,2^{K}-2\}\\
    \varphi^{root} &= R(\mathcal{A},\mathcal{L},U) = \mathcal{A}^s(\Delta_U^{\mathcal{L}^s})\\
    H^{root} &= \verb|hash|\big(\mathcal{A}^s(\Delta_U^{\mathcal{L}^s})\big)
\end{aligned}
\]
\end{definition}


In Definition~\ref{def:csmt}, $\varphi^k_P$ and $H^k_P$ denote, respectively, the aggregated value and the corresponding hash at parent node $P$ on level~$l$ of the tree. The leaf values $\varphi^0_j$ are derived through the salted transformation $\mathcal{L}^{s}$ over user data $(\delta_i, \mu_i, \tau_i)$, where $\verb|Bin|(j)$ represents the binary encoding of the leaf position.  Empty leaves are assigned the default salted value $\mathcal{L}^{s}(\varnothing)$. Specifically, for occupied leaves, the SMT uses the full $K$-bit output without any truncation. For any input $x$, the leaf index is derived as $\mathcal{N}(x) = \mathsf{Decimal}(\mathsf{hash}(x))$, where $\mathsf{Decimal}(\cdot)$ interprets the $K$-bit string as an integer in $\{0, \ldots, 2^K - 1\}$. On the other hand, for the default leaves $\mathcal{L}^s(\varnothing)$ is a fixed, publicly known constant established at setup. As a result, any non-empty leaf satisfies $\varphi^0_j = \mathcal{L}^s(\delta_i, \mu_i, \tau_i) \neq \mathcal{L}^s(\varnothing)$ by construction, since it encodes real user data. As a result, the root hash $H^{\text{root}}$ serves as a global cryptographic commitment to the aggregation outcome $\varphi^{root}$ which is nothing but the global reduction $\mathcal{A}^{s}(\Delta^{\mathcal{L}^{s}}_U) = R(\mathcal{A},\mathcal{L},U)$ for all users. Based on the definition of CSMT, we can define the Merkle consistency as defined in Property \ref{prop:mer_con}.  

\begin{figure}[t]
\centering
\resizebox{0.35\textwidth}{!}{
\begin{tikzpicture}[
  font=\scriptsize, 
  arrow/.style={-Latex, line width=0.5pt},
  box/.style={draw, rounded corners, align=center, inner sep=1.5mm, minimum height=6mm},
  note/.style={align=center},
  aggblocki/.style={
    draw, rounded corners, align=center,
    fill=yellow!20, inner sep=0.5mm, minimum height=6mm,
    font=\scriptsize\itshape
  }
]

\node[note, font=\small] (mrlabel) {Merkle Root};

\node[box, fill=violet!20, draw=violet!60!black,
      below=4mm of mrlabel] (Hroot)
{$H_{\mathrm{root}}$};

\node[box, fill=violet!10, draw=violet!60!black,
      below=4mm of Hroot] (phiroot)
{$\phi_{\mathrm{root}} = A_s(\Delta_{L_s}^{U})$};

\draw[arrow] (phiroot) -- (Hroot);

\node[box, fill=yellow!25, draw=yellow!40!black,
      minimum width=30mm, below=6mm of phiroot] (aggblock)
{\textbf{Recursive Aggregation}\\
$H_P^k=\mathrm{hash}(H_L^k,H_R^k)$};

\draw[arrow] (aggblock) -- (phiroot);

\node[box, fill=yellow!35, draw=yellow!55!black,
      below=5mm of aggblock, xshift=-20mm] (Hi)
{$h_i=\mathrm{hash}(\mathcal{A}_i)$};

\node[box, fill=yellow!35, draw=yellow!55!black,
      below=5mm of aggblock] (Hj)
{$h_j=\mathrm{hash}(\mathcal{A}_j)$};

\node[box, fill=yellow!25, draw=yellow!60!black,
      below=5mm of aggblock, xshift=20mm] (Hk)
{$h_k=\mathrm{hash}(\mathcal{A}_k)$};

\draw[arrow] (Hi) -- (aggblock);
\draw[arrow] (Hj) -- (aggblock);
\draw[arrow] (Hk) -- (aggblock);

\node[aggblocki, below=4mm of Hi] (Ai)
{$\mathcal{A}_i=\mathcal{A}^{\ell}(\phi^0_i,\phi^0_{\emptyset}$)};

\node[box, fill=green!20, draw=green!50!black,
      below=4mm of Ai, xshift=-6mm] (Leafi)
{$\phi^0_i$};

\node[box, fill=red!20, draw=red!60!black,
      below=4mm of Ai, xshift=6mm] (LeafDi)
{$\phi^0_{\emptyset}$};

\draw[arrow] (Leafi.north) -- (Ai.south);
\draw[arrow] (LeafDi.north) -- (Ai.south);
\draw[arrow] (Ai.north) -- (Hi.south);

\node[aggblocki, below=4mm of Hj] (Aj)
{$\mathcal{A}_j=\mathcal{A}^{\ell}(\phi^0_j,\phi^0_{\emptyset}$)};

\node[box, fill=green!20, draw=green!50!black,
      below=4mm of Aj, xshift=-6mm] (Leafj)
{$\phi^0_j$};

\node[box, fill=red!20, draw=red!60!black,
      below=4mm of Aj, xshift=6mm] (LeafDj)
{$\phi^0_{\emptyset}$};

\draw[arrow] (Leafj.north) -- (Aj.south);
\draw[arrow] (LeafDj.north) -- (Aj.south);
\draw[arrow] (Aj.north) -- (Hj.south);

\node[aggblocki, below=4mm of Hk] (Ak)
{$\mathcal{A}_k=\mathcal{A}^{\ell}(\phi^0_k,\phi^0_{\emptyset}$)};

\node[box, fill=green!20, draw=green!50!black,
      below=4mm of Ak, xshift=-6mm] (Leafk)
{$\phi^0_k$};

\node[box, fill=red!20, draw=red!60!black,
      below=4mm of Ak, xshift=6mm] (LeafDk)
{$\phi^0_{\emptyset}$};

\draw[arrow] (Leafk.north) -- (Ak.south);
\draw[arrow] (LeafDk.north) -- (Ak.south);
\draw[arrow] (Ak.north) -- (Hk.south);

\node[box, fill=blue!15, draw=blue!60!black,
      minimum width=10mm, below=3mm of Leafi] (Xi)
{$[(\delta_i,\mu_i),\tau_i]$};

\node[box, fill=gray!20, draw=gray!70!black,
      minimum width=8mm, below=3mm of Xi] (Fi)
{$F(\delta_i,\mu_i)$};

\node[box, fill=blue!10, draw=black!60,
      minimum width=8mm, below=3mm of Fi] (Ci)
{$\mathrm{Concat}(\delta_i,\mu_i)$};

\node[box, fill=orange!20, draw=black!70,
      minimum width=8mm, below=3mm of Ci] (di)
{$\delta_i$};

\draw[arrow] (Xi) -- (Leafi);
\draw[arrow] (Fi) -- (Xi);
\draw[arrow] (Ci) -- (Fi);
\draw[arrow] (di) -- (Ci);

\node[box, fill=blue!15, draw=blue!60!black,
      minimum width=10mm, below=3mm of Leafj] (Xj)
{$[(\delta_j,\mu_j),\tau_j]$};

\node[box, fill=gray!20, draw=gray!70!black,
      minimum width=8mm, below=3mm of Xj] (Fj)
{$F(\delta_j,\mu_j)$};

\node[box, fill=blue!10, draw=black!60,
      minimum width=8mm, below=3mm of Fj] (Cj)
{$\mathrm{Concat}(\delta_j,\mu_j)$};

\node[box, fill=orange!20, draw=black!70,
      minimum width=8mm, below=3mm of Cj] (dj)
{$\delta_j$};

\draw[arrow] (Xj) -- (Leafj);
\draw[arrow] (Fj) -- (Xj);
\draw[arrow] (Cj) -- (Fj);
\draw[arrow] (dj) -- (Cj);

\node[box, fill=blue!15, draw=blue!60!black,
      minimum width=10mm, below=3mm of Leafk] (Xk)
{$[(\delta_k,\mu_k),\tau_k]$};

\node[box, fill=gray!20, draw=gray!70!black,
      minimum width=8mm, below=3mm of Xk] (Fk)
{$F(\delta_k,\mu_k)$};

\node[box, fill=blue!10, draw=black!60,
      minimum width=8mm, below=3mm of Fk] (Ck)
{$\mathrm{Concat}(\delta_k,\mu_k)$};

\node[box, fill=orange!20, draw=black!70,
      minimum width=8mm, below=3mm of Ck] (dk)
{$\delta_k$};

\draw[arrow] (Xk) -- (Leafk);
\draw[arrow] (Fk) -- (Xk);
\draw[arrow] (Ck) -- (Fk);
\draw[arrow] (dk) -- (Ck);

\end{tikzpicture}%
} 
\caption{CSMT leaf-level aggregation and hashing.}
\label{fig:csmt_leaf}
\end{figure}
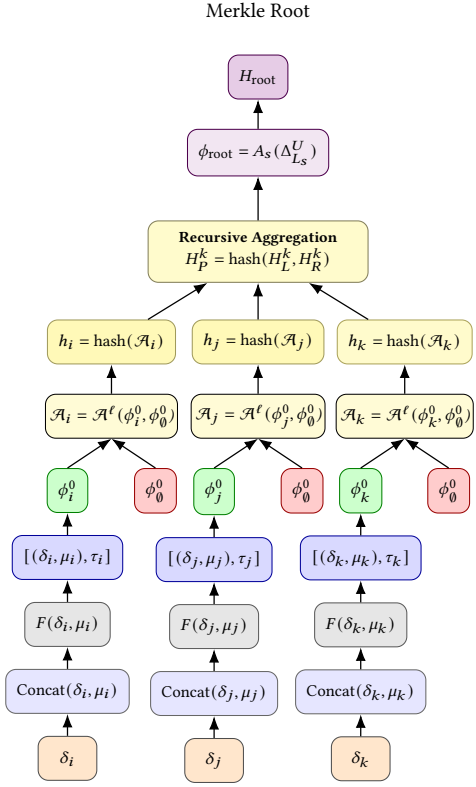

The CSMT supports any local aggregation function $\mathcal{A}^l(\varphi_L, \varphi_R; \theta)$ that is \emph{commutative} and \emph{associative}. This is directly analogous to the class of reduction based computation supported by distributed memory oriented collectives adopted by high-performance frameworks such as MPI \cite{dongarra1995introduction}. Concretely, element-wise sums, products, minimum, maximum, and the sufficient statistics required by the statistical pipelines evaluated in this work decompose naturally into such reductions. Functions outside this class include \emph{order-sensitive} operations such as median and rank statistics and \emph{cross-leaf predicates} needing pairwise comparison across the entire datasets. Supporting such functions would require auxiliary circuit structure beyond the current CSMT design and is left as future work. 

\begin{property}[Merkle Consistency]\label{prop:mer_con}
A CSMT is deemed to be consistent if the following conditions apply
\[
H^k_{L} = H^{k-1}_{L,P}\ \text{ and }\ H^k_{R} = H^{k-1}_{R,P}
\]
where $H^{k-1}_{L,P},H^k_{R} = H^{k-1}_{R,P}$ denote the hash of the left and right subtree roots and $H^k_L$ and $H^k_R$ denote the children node hashes at tree height $k$.
\end{property}

As a direct consequence of Definition~\ref{def:csmt}, each leaf in the CSMT admits a unique structural index determined by its salted hash value, formalized in Property ~\ref{prop:leaf_path}.
\begin{property}[CSMT Inclusion]\label{prop:leaf_identity}
Given a data tuple $(\delta,\mu,\tau)$ for user $u \in U$, its corresponding leaf node index in the CSMT denoted by $\mathcal{N}_u$ can be computed using the salted leaf transform $\mathcal{L}^s$ as follows.
\[
\mathcal{N}_u \;=\; \verb|Decimal|\!\Big[\verb|hash|\big(\mathcal{L}^s(\delta,\mu,\tau)\big)\Big].
\]
\end{property}
\noindent In other words, Property \ref{prop:leaf_identity} ensures that each user can be assigned to a unique leaf index $\mathcal{N}_u$ based on the decimal representation of the hash corresponding to their salted leaf transform.

\begin{property}[CSMT Exclusion]\label{prop:leaf_unoccupied}
Given a leaf node hash $H$ such that 
\[
\verb|Decimal|(H) = \mathcal{N} \text{ and } \varphi^0_{\mathcal{N}} = \mathcal{L}^s(\varnothing)
\]
then there exists no user $u \in U$ with data tuple corresponding to $(\delta,\mu,\tau)$ that leads to $\verb|hash|(\mathcal{L}^s(\delta,\mu,\tau))=H$, for any combination of $\delta\in \mathbb{R}^d,\mu\in\mathbb{R}^{s_u},\tau\in\mathbb{R}^{s_t}$. 
\end{property}
\noindent In simpler terms, Property \ref{prop:leaf_unoccupied} states that a leaf index that maps to the default salted element represents an unoccupied position in the tree. 

\begin{property}[CSMT Path Derivation]\label{prop:leaf_path}
Given a leaf node index $\mathcal{N}$, with the corresponding binary string $\verb|Bin|(\mathcal{N})= [b_1,b_2,\ldots,b_K]\in \{0,1\}^K$ of size $K$, the path from the CSMT root (denoted by level $K$) to the $\mathcal{N}^{th}$ leaf node index can be provided based on the following conditions $\forall k\in\{K,K-1,\ldots 1\}$:
\[
b_{K-l+1} = 
          \begin{cases}
              0, & \text{move to the left child node},\\[4pt]
              1, & \text{move to the right child node.}
          \end{cases}
\]
\end{property}
\noindent Property \ref{prop:leaf_path} exploits the binary representation of $\mathcal{N}_u$ to derive a path from the root to the leaf node index. More precisely, the path derivation scheme adopts a convention wherein the successive child nodes are determined recursively by examining the current bit value. 

Because leaf indices are derived as $\mathcal{N} = \mathsf{Decimal}(\mathsf{hash}(\varphi^0_j))$ and parent node hashes are computed as $H^k_P = \mathsf{hash}(\mathcal{A}^l(\varphi^k_L, \varphi^k_R;\theta))$, any attempt by a computationally bounded adversary to produce a non-default leaf $\varphi^0_j \neq \mathcal{L}^s(\varnothing)$ whose aggregation path yields the same root hash as a default leaf requires finding a collision under \verb|hash|. Under the collision resistance property of Definition~\ref{defn2}, this occurs with at most negligible probability. Therefore, non-membership proofs are unambiguous: a leaf position containing $\mathcal{L}^s(\varnothing)$ cannot be computationally indistinguishable from one containing a valid participant record, and default leaf uniqueness reduces entirely to the collision resistance of the underlying hash function without requiring any additional assumptions on $\mathcal{A}^l$. 

As a result, combining Properties \ref{prop:leaf_identity}, \ref{prop:leaf_unoccupied} and in addition to the one-wayness and collision resistance properties of $\verb|hash|$, we can also state that valid Merkle paths are also consistent according to Property \ref{prop:mer_con}. Without loss of generality, we denote bit values 0,1 to represent left and right child nodes respectively.

%
%
\begin{proposition}\label{prs:csmt}
Given a user $u\in U$ with a data tuple $(\delta,\mu,\tau)$, the CSMT inclusion and exclusion properties form necessary and sufficient conditions for demonstrating membership and non-membership of user $u$ with respect to membership function $\mathcal{M}([u,\delta]|R)$, where $R = \mathcal{A}\Big(\Delta^{\mathcal{L}}_U;\theta_{\mathcal{A}}\Big)$ denotes the global aggregated reduction across all users while $\mathcal{A},\mathcal{L}$ represents the unsalted aggregation and leaf transform functions respectively.
\end{proposition}
\begin{proof}
Proof provided in Appendix \ref{subsec:proof_csmt}
\end{proof}

Proposition \ref{prs:csmt} formally shows that a CSMT can effectively help prove whether a particular user's data record was included in a specific computational aggregation step of the larger statistical method driving the clinical study. In other words, Proposition \ref{prs:csmt} can help pinpoint the exact computational step where a user's data record has been used to drive the statistical method that underpins the overarching clinical research. Therefore, Proposition \ref{prs:csmt} has a foundational implication for the \cosmetic framework since it paves the way for a publicly verifiable zero-knowledge driven mechanism for reliably establishing computational membership claim. 





\subsection{Zero Knowledge Assertions for CSMTs}
We present the preliminaries pertaining to the zero knowledge encapsulations of the Computational Sparse Merkle Tree formulations. More specifically, we consider a zkSNARK framework and summarize the foundational steps of setup, prove and verify functions imminent in zkSNARKs. Our treatment for CSMTs includes both layers of CSMT pertaining to salted leaf transforms $\mathcal{L}^s$ as well as layered salted aggregation function $\mathcal{A}^l$.

\begin{definition}[zk-SNARK Setup Phase]\label{defn3}
Given a standard security parameter $\lambda$, the zk-SNARK setup phase is defined by the function \verb|Setup| for both the salted leaf transformation and aggregation functions as follows:
\[
\begin{aligned}
\verb|Setup|(1^\lambda, \mathcal{L}^s, \theta_{\mathcal{L}^s}) &\mapsto (pk_{\mathcal{L}^s}, vk_{\mathcal{L}^s}),\\
\verb|Setup|(1^\lambda, \mathcal{A}^{l}, \theta_{\mathcal{A}^{l}}) &\mapsto (pk_{\mathcal{A}^{l}}, vk_{\mathcal{A}^{l}}).
\end{aligned}
\]
\end{definition}

The setup phase in Definition~\ref{defn3} generates the proving and verification key pairs 
$(pk, vk)$ corresponding to the functions $\mathcal{L}^s$ and $\mathcal{A}^{l}$. 
This phase is typically executed once per circuit instantiation and is parameterized by the model weights $\theta$. 
Since the underlying state-space model $\mathcal{M}$ is pretrained and stable, the setup need not be repeated frequently. 
However, any modification to $\theta$—such as retraining or parameter updates—necessitates re-execution of the setup phase.

\begin{definition}[zk-SNARK Proving Phase]\label{defn4}
The zk-SNARK proof generation phase is defined by the function \verb|Prove| as follows:
\[
\begin{aligned}
\verb|Prove|(pk_{\mathcal{L}^s}, \theta_{\mathcal{L}^s}, \mathcal{L}^s, [\delta, \mu], \mathcal{L}^s(\delta,\mu,\tau))
    &\mapsto \Pi^{u}_{\mathcal{L}^s},\\[4pt]
\verb|Prove|(pk_{\mathcal{A}^l}, \theta_{\mathcal{A}^l}, \mathcal{A}^l, [\varphi^k_L, \varphi^k_R], \mathcal{A}^l(\varphi^k_L, \varphi^k_R))
    &\mapsto \Pi^{k}_{\mathcal{A}^l}.
\end{aligned}
\]
\end{definition}

Definition~\ref{defn4} specifies the proof generation step. 
The \verb|Prove| function takes as input the proving key $pk$, the function parameters $\theta$, and the corresponding witness–statement pairs
$([\delta,\mu], \mathcal{L}^s(\delta,\mu,\tau))$ for the leaf transformation, or 
$([\varphi^k_L,\varphi^k_R], \mathcal{A}^l(\varphi^k_L,\varphi^k_R))$ for the aggregation step.
It outputs a proof artifact $\Pi$, which attests to the correct evaluation of the underlying function in zero knowledge. 
Proofs are generated independently at each layer or recursion depth of the CSMT.

\begin{definition}[zk-SNARK Verification Phase]\label{defn5}
The verification phase of the zk-SNARK protocol is defined by the function \verb|Verify| as follows:

\[
\begin{aligned}
\verb|Verify|(vk_{\mathcal{L}^s}, \Pi_{\mathcal{L}^s}^u)
\mapsto \Phi_{\mathcal{L}^s}^u, \text{ and } \
\verb|Verify|(vk_{\mathcal{A}^{l}}, \Pi_{\mathcal{A}^{l}}^k)
\mapsto \Phi_{\mathcal{A}^{l}}^{u,k}.
\end{aligned}
\]

\end{definition}

Definition~\ref{defn5} describes the verification step, in which a verifier uses the verification key $vk$ to validate a proof $\Pi$ corresponding to a given function execution. 
The \verb|Verify| function outputs a Boolean flag $\Phi \in \{0,1\}$ indicating whether the proof is valid. 
In practice, this allows regulators or auditors to confirm that each reported computation—whether at the leaf or aggregation level—was performed correctly, without requiring access to the underlying private data. For notational simplicity, we refer to leaf transform and Merkle path proofs by the acronyms LTR and MRP proofs respectively. As a consequence of Definition \ref{defn4}, we collate the set of LTR and MRP proof artifacts into a distinct proof tuple denoted by $\Pi^u_{\mathsf{CSMT}}$ for every user as denoted by Equation \eqref{eq:collated}.
\begin{gather}\label{eq:collated}
    \text{LTR Proof}: \Pi^u_{\mathcal{L}^s}, \
    \text{MRP Proof Set}:\Pi^{u,1:K}_{\mathcal{A}^l} = \big\{\Pi^{u,1}_{\mathcal{A}^l},\ldots,\Pi^{u,K}_{\mathcal{A}^l}\big\}\\
    \text{CSMT Proof Set}:\Pi^u_{\mathsf{CSMT}} = \Big[\Pi^u_{\mathcal{L}^s}, \Pi^{u,1:K}_{\mathcal{A}^l}\Big]
\end{gather}

\begin{proposition}\label{prs:pi_mem}
Given a user $u$ and individual datum $\delta$, an aggregated reduction value $\mathcal{R}(\mathcal{A},\mathcal{L},U)$ and the proof tuple set $\Pi^u_{\mathsf{CSMT}}$, and consistent $vk_{\mathcal{L}^s},vk_{\mathcal{A}^l}$, the following conditions are necessary and sufficient for realizing the membership function $\mathcal{M}([u,\delta]|\mathcal{R}(\mathcal{A},\mathcal{L},U))$
\[
\begin{aligned}
\Phi_{\mathcal{L}^s}^u = 1 \text{ and }\ \Phi_{\mathcal{A}^{l}}^{u,k} = 1, \quad \forall k\in\{1,K\}
\end{aligned}
\]
\end{proposition}

\begin{proof}
Proof given in Appendix \ref{subsec:proof_pi_mem}
\end{proof}

Proposition \ref{prs:pi_mem} shows that the successful public verification of each individual zk-SNARK artifacts in the set $\Pi^u_{\mathsf{CSMT}}$ can help demonstrate the utilization (or lack thereof), of a user data record to drive an individual aggregation step that is implemented through a CSMT. At a fundamental level, Proposition \ref{prs:pi_mem} helps realize the implications of Proposition \ref{prs:csmt} purely in terms of zk-SNARKs. Proposition \ref{prs:pi_mem} ensures that the set of zk-SNARKs contained in $\Pi^u_{\mathsf{CSMT}}$ serve as the inclusion or exclusion guarantees for user datasets.




\section{Algorithmic Foundations of \cosmetic}
We construct the algorithmic foundation of the \cosmetic framework by discussing the set of infrastructure assumptions, followed by the prover and verifier algorithmic components.

\subsection{Assumptions}
The assumptions for the \cosmetic framework encompass both endogenous as well as exogenous factors pertaining to data storage and trusted setups.
\begin{assumption}[PHR Database]\label{as:phr}
    The raw data record and user salt $\delta_u,\mu_u$ along with a unique transform salt $\tau_u$ of each user $u$ are part of a personal health record (PHR) database capable of providing publicly verifiable membership proofs of individual user datum.
\end{assumption}
Assumption \ref{as:phr} postulates the existence of a personal health record (PHR) database that stores and manages raw data record, user and transform salt. Additionally, we also assume that the PHR database is capable of providing Merkle membership proofs of $(\mathsf{hash}(\delta_u,\mu_u),\mathsf{hash}(\tau_u))$ through simple Merkle trees. The primary advantage of doing so would be to prevent data tampering or misuse of user data records in the clinical study itself. \cosmetic also supports multiple such PHR databases that store data for different users participating in a single study, provided that each database can substantiate user records through simple Merkle proofs. 

\begin{assumption}[Trusted Environment]\label{as:ts}
A trustworthy, secure environment exists which:
\begin{enumerate}
    \item guarantees the existence of a Common Reference String (CRS) for zk-SNARK generation, through a publicly verifiable ceremony.
    \item generates proving and verification keys, as part of the zk-SNARK setup phase.
\end{enumerate}
\end{assumption}
Assumption \ref{as:ts} discuses the presence of a trustworthy environment where the proving and verification keys for zk-SNARK circuits are generated as part of the setup phase. The implementation of \cosmetic is based on the \texttt{ezkl} framework \cite{ezkl2024} which uses the Halo2 proof generation backend. As a result, our proof generation framework ultimately relies on the KZG commitment scheme with a universal Powers of Tau ceremony.  



\begin{assumption}[Clinical Research Organization (CRO)]\label{as:cro}
    A CRO exists that governs the implementation of the study by: 
    \begin{enumerate}
        \item acquiring user data from the PHR database. 
        \item delivering CSMT zk-SNARK artifacts with LTR and MRP proofs.
        \item publicly disclosing MRP and LTR proofs for each user. 
    \end{enumerate}
\end{assumption}

Assumption \ref{as:cro} assumes the existence of a clinical research organization (CRO) which is the driver of the entire clinical research study. The CRO is assumed to handle the acquisition of user data and corresponding transform salts of each user from the PHR database; as well as generating and disclosing proof artifacts. It is not necessary for the the CRO and the trusted environment assumed in Assumption \ref{as:ts} to be identical. In fact, the \cosmetic architecture allows these two entities to be distinct in the real-world. As a result, scenarios wherein a third party, such as a regulatory authority, generates and delivers the zk-SNARK circuits for the CRO to carry out the clinical study is plausible. Therefore, based on Assumptions \ref{as:phr}, \ref{as:ts} and \ref{as:cro}, we now discuss the algorithm design of the \cosmetic architecture.

\subsection{Adversarial Model}\label{sec:adv_model}

We now formally characterize the adversarial capabilities assumed in \cosmetic and delineate the boundaries of its cryptographic guarantees.

\noindent\textit{CRO Adversarial Capabilities}: The CRO is modeled as a \emph{computationally bounded malicious adversary}. We assume that the CRO may adaptively select the included participant set $U^{inc} \subseteq U$ after observing study outcomes; selectively omit records from $U^{inc}$ to bias statistical results; and attempt to substitute a participant's data under a different identity. The CRO \emph{cannot}, however, forge valid zk-SNARK proofs which would require breaking the knowledge soundness of the underlying proof system (see Proposition~\ref{prs:ks_lrt_mrp}). Collusion between the CRO and the trusted setup party (Assumption~\ref{as:ts}) is considered out of scope. If the CRS is adversarially generated, soundness guarantees for any zk-SNARK-based system collapse and are therefore treated as a residual trust assumption common to the class of KZG-based constructions.

\noindent\textit{PHR Database as a Trust Boundary}:
The PHR database (Assumption~\ref{as:phr}) is treated as a \emph{declared trust boundary}. We assume that correctness of identity binding that establishes leaf index as a derivation of $(\delta_u, \mu_u, \tau_u)$ corresponds to the genuine data of user $u$ if the PHR database is itself uncompromised. A malicious PHR database could facilitate cross-identity substitutions (attaching user $u$'s data under a different identity $u'$). We acknowledge this as a fundamental limitation of the architecture; mitigating it would require a higher-assurance PHR infrastructure outside the scope of \cosmetic.

\noindent\textit{Two-Layer Completeness Protection}:
While \cosmetic cannot cryptographically enforce \emph{dataset completeness} in full generality, the framework provides two complementary protections against the most operationally relevant completeness attacks. First at the injection layer, every LTR proof generated by the CRO (Algorithm~\ref{alg:leaf_transform_proof_rewritten}) chains to a PHR membership proof via Assumption~\ref{as:phr}. A regulator can therefore demand PHR-backed proofs for every non-default leaf in the published CSMT, and flag any leaf whose claimed identity cannot be corroborated by the PHR database. This prevents a malicious CRO from injecting fabricated participants or substituting one user's data under a different identity. Next at the Omission Detection Layer, a regulator suspecting selective omission of participants with unfavorable outcomes can issue a challenge set $U^{chall} \subseteq U$. The CRO is compelled to rerun \textsc{BuildSMT} over $U^{inc} \cup U^{chall}$ using the \emph{existing, already-published} verification keys without requiring a new setup. A statistically significant discrepancy between the original and challenged aggregate roots is evidence of omission.

\noindent\textit{Scope of Guarantees}:
\cosmetic provides \emph{computational correctness} guarantees which ensure that given that the CRO supplies a dataset, the zk-SNARK proofs certify that the statistical pipeline was executed faithfully over that dataset. \cosmetic does \emph{not} provide \emph{dataset completeness} guarantees in the information-theoretic sense. In other words, it cannot certify that all eligible participants were included. However, it can be used in conjunction with a trusted PHR database to verify data membership while serving as the foundation for regulatory inquiries for various user subsets. Similarly, from an information theoretic perspective, \cosmetic does not guarantee \emph{liveness}, implying that the CRO can theoretically ignore proof requests from particular participants. These limitations are consistent with the broader class of verifiable computation systems and are best addressed through complementary organizational or regulatory mechanisms.

Collectively using the Assumptions \ref{as:phr}, \ref{as:ts} and \ref{as:cro}, as well as our adversarial threat model we now discuss the algorithm design of the \cosmetic architecture.

\subsection{CRO Oriented Algorithmic Components}
The CRO side algorithmic components broadly handles tasks pertaining to leaf transformations, CSMT construction and generation of corresponding zk-SNARKS. For brevity, we abstract away operations pertaining to private data storage on the CRO side.
\subsubsection{Salted Leaf Transforms}:
In Algorithm \ref{alg:leaftransform}, we present the function $\verb|LeafTransform|$ which is denotes the LTR operation handling the transformation of raw user data records including user and transformation salts. First, the function generates the salted transform leaf value $\varphi^0$ and its associated witness $\Omega_{LT}$ based on the compiled circuit for leaf transform $\mathcal{L}^s$ which is parametrized by $\theta_{\mathcal{L}^s}$. The function stores the witnesses privately indexed by the hash of salted raw data tuple, the transform salt and the choice of LTR circuit. Finally, the function returns the transformed value $\varphi^0$, the hash of the leaf transform $H^{leaf}$, transform salt hash $H^{\tau}$ and the leaf index $\mathcal{N}$.
\begin{algorithm}[htbp]
 \caption{Function for Salted Leaf Transformation}\label{alg:leaftransform}
 \begin{algorithmic}[1]
    \Function{LeafTransform}{$\delta,\mu,\tau,\mathcal{L}^s,\theta_{\mathcal{L}^s}$}
        \State compute salted leaf transform $\varphi^{0} \leftarrow 
            \mathcal{L}^{s}([\delta,\mu,\tau];\theta_{\mathcal{L}^s})$
        \State set $H^{(\delta,\mu)}\leftarrow\mathsf{hash(\delta,\mu)}$ and $H^{\tau}\leftarrow\mathsf{hash}(\tau)$
        \State generate witness $\Omega_{LT}\leftarrow[\delta,\mu,\tau, \varphi^{0}]$
        \State\LineComment{store LT witness privately}
        \State \Call{StorePrivateLTWitness}{$[H^{(\delta,\mu)},H^{\tau},\mathcal{L}^s],[\Omega_{LT},\theta_{\mathcal{L}^s}]$}
        \Statex\LineComment{determine leaf hash and index}
        \State set $H^{leaf}\leftarrow\mathsf{hash}(\varphi^{0})$
        \State set $\mathcal{N}\leftarrow\mathsf{Decimal}(H^{leaf})$

        
        \State \Return $\varphi^{0}$, $H^{leaf}$, $H^{\tau}$, $\mathcal{N}$ 
    \EndFunction
 \end{algorithmic}
\end{algorithm}

\subsubsection{CSMT Construction}: 
The CRO constructs the CSMT using the function $\verb|BuildSMT|$ as represented in Algorithm \ref{alg:buildsmt}. The function consumes a given user set $U$ with salted leaf transformations and indices denoted by $\{\mathcal{N}_i,\varphi^{0}_i\}_{u_i\in U}$, a CSMT tree height of $K$, as well as the aggregation function $\mathcal{A}^l$ parametrized by $\theta_{\mathcal{A}^l}$. The CRO inserts the non-default user leaves at the appropriate locations while the rest are left with the default value. Using a bottom-up aggregation approach, the CRO builds the sparse Merkle tree by recursively computing the parent for each node identified by its hash using $\mathcal{A}^l(\cdot,\theta_{\mathcal{A}^l})$. The aggregation results in the root value $\Psi^{K}[0]$, root hash $H^{K}[0]$ and CSMT witnesses arranged in a sparse Merkle tree format denoted by $\Omega_{CSMT}$. While the CSMT witnesses are stored privates, the function returns the root value and the root hash as their outputs. 
\begin{algorithm}[htbp]
 \caption{CSMT Construction} 
 \label{alg:buildsmt}
 \begin{algorithmic}[1]

    \Function{BuildSMT}{$\{\mathcal{N}_i,\varphi^{0}_i\}_{u_i\in U},
                        K,\mathcal{A}^{l},
                        \theta_{\mathcal{A}^l}$} 
        \State set $\Psi^{0}[j] \leftarrow \mathcal{L}^s(\varnothing)$ $\forall j \in \{0, 2^K-1\}$ \Comment{Initialize $2^K$ leaves}
        \Statex\LineComment{insert valid user leaves at their hashed indices}
        \State $\Psi^{0}[\mathcal{N}_i] \leftarrow \varphi^{0}_i$ $\forall (\mathcal{N}_i,\varphi^{0}_i)$
        \Statex \LineComment{compute leaf-level hashes}
        \State $H^{0}[j] \leftarrow \verb|hash|(\Psi^{0}[j])$ $\forall j \in \{0, 2^K-1\}$
        \Statex \LineComment{bottom-up aggregation}
        \For{$k = 1$ \textbf{to} $K$}
            \For{$j = 0$ \textbf{to} $2^{K-k}-1$}
                \State $\varphi^{k}_L \leftarrow \Psi^{k-1}[2j]$
                \State $\varphi^{k}_R \leftarrow \Psi^{k-1}[2j+1]$
                \State compute parent aggregate $\Psi^{k}[j] \leftarrow 
                    \mathcal{A}^{l}(\varphi^{k}_L,\varphi^{k}_R;\theta_{\mathcal{A}^{l}})$
                \State compute parent hash $H^{k}[j] \leftarrow \verb|hash|(\Psi^{k}[j])$
            \EndFor
        \EndFor
        \Statex \LineComment{construct CSMT Witness object for membership tests}
        \State $\Omega_{CSMT} \leftarrow 
               \Big[\big\{\Psi^{k}[j]\big\}_{j=0}^{2^{K-k}-1}\Big]_{k=0}^{K}$
        \State \Call{StorePrivateCSMTWitness}{$\mathcal{A}^l,(\Omega_{CSMT},\theta_{\mathcal{A}^l})$}
        \Statex \LineComment{return global aggregate, root hash and CSMT witness}
        \State \Return $\Psi^{K}[0]$, $H^{K}[0]$ 
        
    \EndFunction
 \end{algorithmic}
\end{algorithm}

\subsubsection{Generation of LTR Proofs}:
In function $\verb|CRO-LTRProve|$ given in Algorithm \ref{alg:leaf_transform_proof_rewritten}, we discuss the mechanism to generate zk-SNARKs for the leaf transformation on raw user data records. The function consumes the hashes of salted raw user data and the transform salts denoted by $H^{(\delta,\mu)},H^{\tau}$ respectively. Additionally the verification key $vk_{\mathcal{L}^s}$ for the specific leaf transformation function must also be provided to help identify the specific transformation function for which the proofs are being requested. Consequently, the function loads the corresponding proving key $pk_{\mathcal{L}^s}$, the compiled circuit $\mathcal{L}^s$. Next, the LTR witnesses and parameters $\Omega_{LT},\theta_{\mathcal{L}^s}$ are looked up based on the provided salted raw user data record and transform salts. As a result, the CRO generates the zk-SNARK $\Pi_{\mathcal{L}^S}$ and returns the leaf hash, index and the SNARK artifact denoted by $H^{leaf},\mathcal{N}$ and $\Pi_{\mathcal{L}^s}$ respectively.

\begin{algorithm}[htbp]
 \caption{Salted Leaf Transform Proof Generation}
 \label{alg:leaf_transform_proof_rewritten}
 \begin{algorithmic}[1]

    \Function{CRO-LTRProve}{$H^{(\delta,\mu)},H^{\tau},vk_{\mathcal{L}^s}$}
    

        \State load compiled circuit $\mathcal{L}^s$ based on verification key $vk_{\mathcal{L}^s}$
        \State load proving keys $pk_{\mathcal{L}^s}$ based on $\mathcal{L}^s$
        \Statex\LineComment{load the LT witness and parameters}
                \State $(\Omega_{LT},\theta_{\mathcal{L}^s})\leftarrow$\Call{LoadPrivateLTWitness}{$H^{(\delta,\mu)},H^{\tau},\mathcal{L}^s$}

        
        \Statex\LineComment{generate LTR zk\text{-}SNARK}
        \State
            $\Pi_{\mathcal{L}^s} \leftarrow
            \verb|Prove|(
                pk_{\mathcal{L}^s},
                \mathcal{L}^s,
                \theta_{\mathcal{L}^s},
                \Omega_{LT}
            )$
        \State determine leaf hash $H^{leaf} \leftarrow \Pi_{\mathcal{L}^s}[\mathsf{output}]$
        \State compute leaf index $\mathcal{N} \leftarrow \verb|Decimal|(H)$
        \State\LineComment{return leaf transform, hash, index and proof artifact}
        \State\Return $H^{leaf}$, $\mathcal{N}$, $\Pi_{\mathcal{L}^s}$

    \EndFunction

 \end{algorithmic}
\end{algorithm}

\subsubsection{Generation of MRP Proofs}:
The mechanism for generating MRP proofs for a given leaf transform of user data is denoted by the function $\verb|CRO-MRPProve|$ which is presented in Algorithm \ref{alg:csmt_path_proof_updated}. The function consumes the leaf hash and index $H^{leaf},\mathcal{N}$, a nonce value $\eta$ supplied by the verifier as well as the verification key $vk_{\mathcal{A}^l}$ to identify the aggregation circuit and parameters. The leaf index is binarized to represent an array of selector or path index bits to help serve as a route from the corresponding leaf to the root of the CSMT. The nonce value plays an important role to help assert the integrity of the per-hop MRP zk-SNARK artifacts with respect to the binary representation of the leaf node index obtained from the LTR proof. 

After loading the proving key,$pk_{\mathcal{A}^l}$, aggregator circuit $\mathcal{A}^l$ and the CSMT witness and parameters $(\Omega_{CSMT},\theta_{\mathcal{A}^l})$, the function iterates through all levels of the tree starting from leaf level. At each iteration, the parent value $\Omega^{P}_{CSMT}$ is obtained from CSMT witness object, and depending on the corresponding path index bit, the orientation (i.e. left or right of the current node) of the sibling is decided. Based on the witnesses, $k^{th}$-hop zk-SNARK is generated which is denoted by $\Pi^{(k)}_{\mathcal{A}^l}$. The function returns the hash of the root $H^{root}$, the series of zk-SNARK artifacts for each hop denoted by $\Pi^u_{\mathsf{CSMT}}$ as well as the array of path index bits $B$.
\begin{algorithm}[htbp]
 \caption{CSMT Merkle Path Proof Generation}
 \label{alg:csmt_path_proof_updated}
 \begin{algorithmic}[1]

    \Function{CRO-MRPProve}{$H^{leaf},\mathcal{N},\eta,vk_{\mathcal{A}^l}$}
        \State initialize empty hop-proof list $\Pi_{\mathsf{CSMT}}$
        
        \State load compiled circuit $\mathcal{A}^l$ based on verification key $vk_{\mathcal{A}^l}$
        \State load proving keys $pk_{\mathcal{A}^l}$ based on $\mathcal{A}^l$

        \Statex\LineComment{load the CSMT witness and parameters}
        \State $(\Omega_{CSMT},\theta_{\mathcal{A}^l})\leftarrow$\Call{LoadPrivateCSMTWitness}{$\mathcal{A}^l$}
        \State compute binary index path $B \leftarrow \verb|Bin|(\mathcal{N})$
        \State set starting index to be $j\leftarrow \mathcal{N}$ which is the leaf index
        \Statex\LineComment{iterate through CSMT levels}
        \For{$k = 0$ \textbf{to} $K$}
            
            \State set parent index $p\leftarrow \mathsf{Floor}(j/2)$
            \State obtain parent node value $\Omega^{P}_{CSMT}\leftarrow\Omega_{CSMT}[p,k]$
            \Statex\LineComment{identify MRP hop witness}
            \If{$B[k]=0$} \Comment{current node is left sibling}
                \State set $\Omega^k_{CSMT} = \{\Omega_{CSMT}[j,k],\Omega_{CSMT}[j+1,k],\Omega^P_{CSMT}\}$
            \Else \Comment{current node is right sibling}
                \State set $\Omega^k_{CSMT} = \{\Omega_{CSMT}[j-1,k],\Omega_{CSMT}[j,k],\Omega^P_{CSMT}\}$
            \EndIf
            \Statex\LineComment{generate MRP hop zk-SNARK}
            \State$\Pi^{(k)}_{\mathcal{A}^l} \leftarrow
            \verb|Prove|(
                pk_{\mathcal{A}^l},
                \mathcal{A}^l,
                \theta_{\mathcal{A}^l},
                \Omega^k_{CSMT},
                B[k],\eta
            )$
            \State append $\Pi^{(k)}_{\mathcal{A}^l}$ to $\Pi_{\mathsf{CSMT}}$
        \EndFor

        \Statex\LineComment{set last hash as CSMT root}
        \State $H^{root} \leftarrow H^k_P$

        \Statex\LineComment{return root hash, CSMT proofs}
        \State\Return $H^{root}$, $\Pi^u_{\mathsf{CSMT}},B$ 

    \EndFunction
 \end{algorithmic}
\end{algorithm}
 
\subsubsection{Integrated \cosmetic CRO Logic}:
We combine the algorithmic components discussed above into a standalone function denoted by $\verb|CosmeticCROBuild|$ as presented in Algorithm \ref{alg:cosmetic_CRO}. Without loss of generality, we consider a PHR database with an exhaustive user base $U$ such that $U^{inc}\subseteq U$ represents the set of users chosen to participate in the user study. The function $\verb|CosmeticCROBuild|$ consumes $U,U^{inc}$ in addition to verification keys $vk_{\mathcal{L}^s},vk_{\mathcal{A}^l}$ for LTR and MRP proofs to identify the relevant aggregation and leaf transform functions respectively. The function acquires the raw salted user data set based on $U$ and carries out a leaf transform using $\verb|LeafTransform|$, assigning a unique randomized transform salt to each user. Consequently, the function invokes $\verb|BuildSMT|$ using the selected group of users $U^{inc}$ to yield the root value and root hash denoted by $\Psi^{root},H^{root}$. Finally, the CRO distributes the hash of transform salt of each user $H^\tau_u$ and the verification keys $vk_{\mathcal{L}^s},vk_{\mathcal{A}^l}$ corresponding to the transform and aggregation functions for all users.  Next, the verification keys for leaf transform and aggregation along with the root hash are also made public. 

To prevent circuit substitution attacks, \cosmetic can also enable regulators to generate pre-approved circuits for common leaf and aggregation operations. Since the verification keys are deterministically derived from the circuit, any verifier can confirm that a regulator approved circuit was used before accepting proofs to prevent circuit substitution attacks. Additionally, any addition or removal of participants only require leaf value and root updates while reusing existing verification keys. A new trusted setup is required only if K itself grows. Since the architecture relies on a specific choice of the hash function with K bits, fluctuating trial enrollment does not require any new zkSNARK setup steps. We note that even though Algorithm \ref{alg:cosmetic_CRO} considers all users in the PHR database for effecting the leaf transforms, doing so is not necessary and is presented as such for representational convenience and to highlight the generation of both inclusion and exclusion proofs in later sections.  




\begin{algorithm}[htbp]
 \caption{CRO Component of \cosmetic}
 \label{alg:cosmetic_CRO}
 \begin{algorithmic}[1]
   \Function{CosmeticCROBuild}{$U,U^{inc},vk_{\mathcal{L}^s},vk_{\mathcal{A}^l}$}
        \State acquire user set $U$, salted raw data $\Delta^s_U$, transform salt array $T_U$ from PHR database
                \Statex
        \LineComment{load LTR and MRP proof generation artifacts}
        \State load compiled circuits $\mathcal{L}^s,\mathcal{A}^l$ and parameters $\theta_{\mathcal{L}^s},\theta_{\mathcal{A}^l}$
        \Statex
        \LineComment{build CSMT}
        \For{$u\in U$}%
        \State acquire $(\delta^u,\mu^u) \in \Delta^s_U$
        \State acquire transform salt $\tau^u$ for user $u$ from $T_U$
        \State $\varphi^{0}_u,H^{leaf}_u,H^{\tau}_u,\mathcal{N}_u\leftarrow$\Call{LeafTransform}{$\delta_u,\mu_u,\tau_u,\mathcal{L}^s,\theta_{\mathcal{L}^s}$}
        \EndFor
        
        \Statex \LineComment{build SMT based on included user set $U^{inc}$ in clinical study}
        \State  $\Psi^{root},H^{root}\leftarrow$ \Call{BuildSMT}{$\{\mathcal{N}_u,\varphi^{0}_u\}_{u\in U^{inc}},
                        K,\mathcal{A}^{l},
                        \theta_{\mathcal{A}^l}$}

        \Statex \LineComment{send transform salt hash and verification keys to users}
        \State distribute $[H^{\tau}_{u},]$ for all users $u\in\text{user set } U$
        \State publish $H^{root},,vk_{\mathcal{L}^s},vk_{\mathcal{A}^l}$ publicly
        \State \Return $\Psi^{root},H^{root}$
   \EndFunction
 \end{algorithmic}
\end{algorithm}
\subsection{Verifier Oriented Algorithmic Components}
\vspace{-2.5mm}
The algorithmic components on the Verifier are in charge of validating zk-SNARKs delivered by the CRO as a means to certify the inclusion of individual user datum consistent with the claimed leaf transform and aggregation functions. We note that the Verifier side logic is purely driven by the exchange of zk-SNARKs and the hash derivatives contained therein. We subdivide the entire task set into three different algorithmic components each pertaining to verifying individual LTR, and per-hop MRP zk-SNARKs, as well as validating inclusion proofs in the CRO CSMT. The same algorithmic structure provides support for both inclusion and exclusion proofs in the \cosmetic algorithmic framework. For brevity, in this section, we only discuss the main verifier logic represented in Algorithms \ref{alg:cosmetic_verifier} and \ref{alg:cosmetic_verifier} that relies on the primitives \texttt{LTRVerify}, \texttt{MRPHopVerify}, which have been discussed at length in Appendix \ref{apx:snarkver}. 

In Algorithm \ref{alg:cosmetic_verinc} we present the \texttt{VerInc} function that relies on the hashes of the salted raw user data record $H^{(\delta_u,\mu_u)}$, the user specific leaf transform $H^{leaf}_u$, CSMT root $H^{root}$ and the nonce $H^\eta$ in addition to the set of LTR and MRP zk-SNARK artifacts and their corresponding verification keys denoted by $\Pi^{u},vk_{\mathcal{L}^s},vk_{\mathcal{A}^l}$ respectively. The function enables the verification of the LTR zk-SNARK using the \texttt{LTRVerify} function. Next, we begin the process of verifying the MRP zk-SNARK artifacts by considering the binary path representation of $H^{leaf}_u$. We iterate over each hop of the CSMT starting from the leaf level with $H^{leaf}_u$ all the way to the root. At each level, we retrieve the right and left input hashes from the per-hop CSMT proof artifact represented by $\Pi_{\mathcal{A}^{l}}^k$. The right and left inputs are checked for consistency with respect to the selector bit corresponding to the current hop as well as CSMT level. Consequently, we validate the zk-SNARK of the hop using the \texttt{MRPHopVerify} function. Passing the validation criteria for every hop as well as the leaf transformation results in a successfully verifying the inclusion of a specific user's data in the clinical study.
We provide Algorithm \ref{alg:cosmetic_verifier} to summarize the algorithmic design at the verifier end. In Algorithm \ref{alg:cosmetic_verifier}, the verifier invokes \texttt{CRO-LTRProve} and \texttt{CRO-MRPProve} functions on the CRO. These CRO based functions can be implemented as an RPC call or a REST API functionality. The verifier then tests the inclusion using the function \texttt{VerInc}. Consistency of the selector bit path and the salted leaf transformation reported by \texttt{CRO-MRPProve} and \texttt{CRO-LTRProve} is critical to ensure that the generated proofs pertain to the same user datum. Additionally a successful outcome of the \texttt{VerInc} function concludes the verification process at the verifier.
\begin{algorithm}[htbp]
 \caption{Function for CSMT Inclusion Verification}
 \label{alg:cosmetic_verinc}
 \begin{algorithmic}[1]
    \Function{VerInc}{$H^{(\delta_u,\mu_u)},H^{\tau},H^{leaf}_u,\Pi^{u},H^{root},H^\eta,vk_{\mathcal{L}^s},vk_{\mathcal{A}^l}$}
        \State extract $\Pi^{u}_{\mathcal{L}^s}$ from $\Pi^{u}$ and $\Pi^{u,1:K}_{\mathcal{A}^l}$ from $\Pi^{u}$
        \State $\Phi_{\mathcal{L}^s}^u\leftarrow$ \Call{LTRVerify}{$vk_{\mathcal{L}^s}, \Pi_{\mathcal{L}^s}^u, H^{(\delta_u,\mu_u)},H^{\tau}, H^{leaf}_u$}
        
        \If{$\Phi_{\mathcal{L}^s}^u = 0$} \Comment{LTR cannot be verified}
        \State \Return False 
        \EndIf
        \State set binary path $B\leftarrow H^{leaf}_u$
        \State set $H^{curr} \leftarrow H^{leaf}_u$
        \For{$k = 1$ \textbf{to} $K$}
            \State retrieve MRP hop proof $\Pi_{\mathcal{A}^{l}}^k \leftarrow \Pi^{u,1:K}_{\mathcal{A}^l}[k]$ for level $k$
                \LineComment{find siblings from MRP hop proof based on selector bit}
                \If{$B[k] = 0$}
                    \State $H^R \leftarrow \Pi_{\mathcal{A}^{l}}^k[\mathsf{RightInput}]$ and $H^L \leftarrow H^{curr}$
                \Else
                   \State $H^L \leftarrow \Pi_{\mathcal{A}^{l}}^k[\mathsf{LeftInput}]$ and $H^R \leftarrow H^{curr}$
                \EndIf

                \Statex\LineComment{verify MRP hop proof}
                \State $\Phi_{\mathcal{A}^{l}}^{u,k} \leftarrow$ \Call{MRPHopVerify}{$vk_{\mathcal{A}^{l}},\Pi_{\mathcal{A}^{l}}^k,H^{L},H^{R},\mathsf{hash}(B[k]),H^{\eta}$}

                \If{$\Phi_{\mathcal{A}^{l}}^{u,k} = 0$} \Comment{MRP cannot be verified}
                \State \Return False
                \EndIf
                \Statex\LineComment{set the MRP hop output as current value}
                \State $H^{curr}\leftarrow \Pi_{\mathcal{A}^s}[\mathsf{Parent}]$
            
        \EndFor
        \LineComment{check root hash and selector path and leaf hash consistency}
        \If{$H^{curr} = H^{root}$ and $\mathsf{Decimal}(B) = H^{leaf}_u$} 
            \State \Return \texttt{True}
        \Else\Comment{inconsistent proving system}
            \State \Return  \texttt{False}
        \EndIf
    \EndFunction
 \end{algorithmic}
\end{algorithm}
\begin{algorithm}[htbp]
 \caption{Verifier Component of \cosmetic}
 \label{alg:cosmetic_verifier}
 \begin{algorithmic}[1]
    \Function{CosmeticVerifier}{u}
    \State load $H^{(\delta,\mu)}$ pertaining to user $u$ from PHR database
    \Statex\LineComment{recieved during proof generation}
    \State load $H^{\tau},vk_{\mathcal{L}^s},vk_{\mathcal{A}^l}$ sent by CRO
    \Statex\LineComment{invoke LTR prover remotely on CRO}
    \State $H^{leaf}$, $\mathcal{N}$, $\Pi_{\mathcal{L}^s}\leftarrow$\Call{CRO-LTRProve}{$H^{(\delta,\mu)},H^{\tau},vk_{\mathcal{L}^s}$}
    \Statex \LineComment{invoke MRP prover remotely on CRO}
    \State $H^{root},\Pi_{\mathcal{A}^l},B\leftarrow$\Call{CRO-MRPProve}{$H^{leaf},\mathcal{N},\eta,vk_{\mathcal{A}^l}$}
    \State consolidate LTR and MRP proofs $\Pi^{u}\leftarrow \{\Pi_{\mathcal{L}^s},\Pi_{\mathcal{A}^l}\}$
    \State $\Phi^{(\delta,\mu)}\leftarrow$\Call{VerInc}{$H^{(\delta,\mu)},H^{leaf},\Pi^{u},H^{root},H^\eta,vk_{\mathcal{L}^s},vk_{\mathcal{A}^l}$}
    \If{$\Phi^{(\delta,\mu)}=0$}
        \State inclusion verification failed
    \EndIf
    \EndFunction
 \end{algorithmic}
\end{algorithm}

\subsection{Security Model and Guarantees}
There are two fundamental formal guarantees that the \cosmetic framework delivers with respect to the CRO. First, using Proposition \ref{prs:csmt}, it establishes the CSMT architecture, providing necessary and sufficient conditions to evaluate the membership of a particular user's data in a reduction operation. Second, Proposition \ref{prs:pi_mem} formally proves that successful verification of LTR and MRP zk-SNARKs are necessary and sufficient to realize the membership function. Lastly, the algorithmic foundations of the \cosmetic framework are driven largely by Assumptions \ref{as:phr} - \ref{as:cro} which serve as necessary conditions for a successful and secure implementation of the framework.

From the algorithmic implementation standpoint, there are exactly two potential gaps where an incorrect evaluation of the membership function cannot be detected solely on the basis of Proposition \ref{prs:pi_mem}. The first gap pertains to the soundness of knowledge argument afforded by zk-SNARKs which inhibits the practical likelihood of obtaining proof artifacts consistent with Proposition \ref{prs:pi_mem} but generated using incorrect witnesses. The second gap focuses on the data exclusivity argument which pertains to the ability to exclusively commit to the set of users whose data was included in a particular reduction operation. The data exclusivity argument is an essential security component to prove the tamper-resistant aspects of the dataset used for the clinical study.

\subsubsection{Knowledge Soundness}: Our knowledge soundness argument is based on the implementation mechanisms for verification highlighted in Algorithm \ref{alg:cosmetic_verifier}. As a consequence, we establish Proposition \ref{prs:ks_lrt_mrp} which provides formal guarantees pertaining to the knowledge soundness of the zk-SNARK artifacts. 
\begin{proposition}\label{prs:ks_lrt_mrp}
The knowledge-soundness property of the underlying zero knowledge system guarantees that the membership (or non-membership) of user $u$ for the reduction operation $R(\mathcal{A},\mathcal{L},U^{inc})$ can be violated with only negligible probability.
\end{proposition}
\begin{proof}
    Proof given in Appendix \ref{subsec:proof_ks_lrt_mrp}
\end{proof}

 Fundamentally, the proposition explores the ability for such a CRO to use incorrect LTR witnesses, inconsistent leaf indices or manipulated MRP proof artifacts for one or more hops as a means to generate zk-SNARKs that can be verified using Algorithm \ref{alg:cosmetic_verifier}. The knowledge-soundness of each LTR and MRP proof instance follows directly from the formally established knowledge-soundness of the Halo2 proof system~\cite{halo2_zcash_2023}, which guarantees that any adversary producing a valid proof must possess a valid witness, except with negligible probability. Therefore a compromised CRO cannot practically generate zk-SNARK artifacts consistent with Proposition~\ref{prs:pi_mem} without access to genuine PHR-authenticated witnesses.  

  The \cosmetic framework employs a deliberate two-level commitment structure wherein the root hash $H^{root} = \mathsf{hash}(\mathcal{A}^s(\Delta^{\mathcal{L}^s}_U))$ commits to the aggregate statistical outcome $R(\mathcal{A}, \mathcal{L}, U)$ rather than to the full dataset. Commitment to individual participants is not derived from the root but from the LTR proofs wherein each leaf hash $H_i = \mathsf{hash}(\mathcal{L}^s(\delta_i, \mu_i, \tau_i))$ is collision-resistant by Definition~\ref{defn2}, and the corresponding zk-SNARK binds that leaf to its PHR-authenticated tuple $(\delta_i, \mu_i, \tau_i)$. The MRP proofs then chain each leaf commitment to the root. Root-level and leaf-level collision resistance therefore serve distinct purposes and should not be treated as the same security property. In other words, the root demonstrates the integrity of the  aggregation, while the leaves prove commitment and chaining to legitimate users who are part of the PHR database. 

\subsubsection{Data Exclusivity}: To demonstrate the data exclusivity argument, we consider $U^{inc}\subseteq U$ as having been the set of users whose data a compromised CRO claims to use for its clinical study while actually utilizing $\hat{U}^{inc}$ with $\hat U^{inc} \cap U^{inc} \neq \varnothing$. We formally define data exclusivity argument in Definition \ref{def:data_exc}. 
\begin{definition}\label{def:data_exc}
    Given membership function $\mathcal{M}$ and a reduction function $R$, a clinical study is deemed to be data exclusive to a set of users $U^{inc}\subseteq U$ if and only if the following conditions hold
    \[
    \nexists u \in U^{inc} \text{ and } \mathcal{M}([u,\delta_u]|R) = 1,
    \]
    where $\delta_u$ is the raw user data record for any user $u$.
\end{definition}
Based on Definition \ref{def:data_exc} we derive necessary conditions for data exclusivity as presented in Proposition \ref{prs:exc_check}. A prerequisite for Proposition \ref{prs:exc_check} is to check the provenance of the leaves by ensuring that every non-default leaf must be the result of a leaf transform on a raw data record that belongs to a PHR database. However, verifying the provenance of each leaf is trivial and straightforward based on Algorithm \ref{alg:csmt_ltr_proof_verify}.
\begin{proposition}\label{prs:exc_check}
Given a set of users $U$ with corresponding ordered salted leaf transformed set denoted by $\Delta^{\mathcal{L}^s}_U$, data exclusivity is guaranteed if the ordered leaf set corresponding to every subtree rooted at a non-default, non-leaf is a subset of $\Delta^{\mathcal{L}^s}_U$. 
\end{proposition}
\begin{proof}
Proof given in Appendix \ref{subsec:proof_exc_check}
\end{proof}

Proposition \ref{prs:ks_lrt_mrp} provides the set of criteria which can be used by the CRO to prove data exclusivity based on analyzing the set of non-default leaves that appear across the entire CSMT. 
We note however that to implement Proposition \ref{prs:exc_check} the CRO would have to divulge the LTR and MRP zk-SNARKs of all users that are part of the set $U^{inc}$ whose data was included in the clinical study. In Appendix \ref{subsec:data_exc_exp}, we provide a detailed discussion for strategies that can successfully detect violations in data exclusivity including Algorithm \ref{alg:csmt_data_exclusivity} which outlines an implementation mechanism for checking the provenance and data exclusivity of a particular reduction operation.  

\section{Experimental Results}
We present the computational experiments for evaluating \cosmetic using three real world clinical case studies that are driven by two distinct clinical datasets. All experiments were conducted on a virtual machine (VM) running Ubuntu~24.04, provisioned with 16 vCPUs and 100~GB of RAM. The implementation was carried out in Python~3.11, with model inference performed using PyTorch~2.7.1. Zero-knowledge proofs were generated using the \texttt{ezkl} library, which underlies the zkSNARK construction within the \cosmetic framework. All experiments were evaluated under multiple fixed-point precision configurations, controlled by the scale parameter. We consider scales of 8, 10, 12, and 14, which directly determine the arithmetic precision inside the zero-knowledge circuits. This range allows us to evaluate the impact of increasing cryptographic precision on performance, numerical stability, and cryptographic overhead. Our source code and documentation can be accessed at \url{https://disys-lab.github.io/cosmetic/}.

\subsection{Clinical Case Studies} \label{sec:5.1}
 
\noindent \textit{Evaluation Scope}: We evaluate \cosmetic across three representative statistical workflows commonly encountered in clinical and genomic studies that pertain to the two-sample Kolmogorov-Smirnov (KS) test, the Logistic Likelihood-Ratio Test (LRT), and the Logistic Accuracy (Acc) test. Our evaluation focuses on three distinct dimensions involving the cost of circuit compilation and witness generation, the stability of the resulting statistical outputs and the cryptographic overhead incurred through proving and verification keys. Together, these metrics characterize the practicality of deploying \cosmetic in privacy-sensitive analytical pipelines. To drive the case studies, we utilize datasets from two real-world cases pertaining to Huntington’s disease (HD) and immunodeficiency virus type 1 (HIV-1).

\noindent\textit{Example Dataset 1: Huntington's disease}. As a genomic example, we use CAG repeat-length data in a two-sample KS test comparing healthy controls and individuals with clinical Huntington's disease. Individual values were simulated from published summary statistics \cite{Moily2014_wn,Jiang2014_we,Gardiner2017_vg,Vater2025_fq}, yielding two cohorts of 12 each. This dataset illustrates distributional testing under CoSMeTIC's inclusion and exclusion guarantees. Brief disease background is given in Appendix \ref{apx:hd}.

\noindent\textit{Example Dataset 2: HIV-1 Temsavir resistance}. As a clinical lab example, we use HIV-1 resistance records from a clinical lab test. CoSMeTIC applies an LRT to compare full and reduced logistic-regression models and an ACC evaluation on a held-out partition. Each sample was associated with the amino acid sequence at four Env positions (375, 426, 434, 475), and a resistance value measured in vitro using the PhenoSense GT assay. After dummy encoding, the feature space expanded to 19 variables.
As part of our study, we examine the Env-targeting therapeutic BMS-626529 (Temsavir, TMR) \cite{Markham2020, Wang2018_xk}. This example illustrates model comparison and predictive evaluation with the same accountability guarantees. Brief clinical background is given in Appendix \ref{apx:hiv}.

For Huntington's we utilize two independent cohorts consisting of 12 individuals each, representing the healthy and diseased (HD) groups. A two-sample KS test is implemented using the \cosmetic framework applied to assess distributional differences between the two groups. In the second case study related to HIV-1, we randomly split the data (n=564) into training (n=479) and test (n=85) sets. For each of the LRT and the ACC tests, a subset of (n=12) participants is used for the experiments with the \cosmetic framework. The architecture of the CSMT within the \cosmetic framework for the KS, LRT and ACC tests has been illustrated in Figure \ref{fig:three_workflows_row} presented in Appendix \ref{apx:algorithms}.
 

\begin{table}[!htb]
\setlength{\tabcolsep}{2pt}
\centering
\begin{tabular}{|l|c|c|}
\hline
\textbf{Step} & \textbf{Included $u$} & \textbf{Excluded $u'$} \\
\hline
Leaf index & $\mathcal{N}_u$ & $\mathcal{N}_{u'}$ \\
Leaf content $\varphi^0_{\mathcal{N}}$ & $\mathcal{L}^s(\delta_u,\mu_u,\tau_u)$ & $\mathcal{L}^s(\varnothing)$ \\
Governing Property & Property ~\ref{prop:leaf_identity} & Property ~\ref{prop:leaf_identity} + \ref{prop:leaf_unoccupied} \\
LTR witness & genuine $(\delta_u,\mu_u,\tau_u)$ & none (default, public) \\
MRP hops $1{:}K$ & pass, all $K$ & pass, all $K$ \\
Root check & $H^{curr}{=}H^{root}$ match & $H^{curr}{=}H^{root}$ match \\
$\mathcal{M}(\cdot\,|\,R)$ & $1$ & $0$ \\
Certifies & membership & non-membership \\
\hline
\end{tabular}
\caption{Worked inclusion and exclusion traces through the CSMT properties underlying Algorithm \ref{alg:cosmetic_verinc}. Both cases share identical path-derivation and aggregation logic; only the leaf content and presence of a genuine LTR witness differ.}
\label{tab:worked_example}
\end{table}

\subsection{Worked Inclusion and Exclusion Example}\label{sec:worked_example}
We illustrate Algorithm \ref{alg:cosmetic_verinc} concretely using two cases drawn from the Huntington's disease KS case study (Section \ref{sec:5.1}) comprising an included participant $u$ of the HD cohort, and an identity $u'$ that was never part of either cohort. Both cases compute the leaf index in the same manner of Property \ref{prop:leaf_identity}, $\mathcal{N} = \text{Decimal}[\text{hash}(\mathcal{L}^s(\delta,\mu,\tau))]$, and the same path-derivation logic of Property \ref{prop:leaf_path} from that index to the root; \cosmetic does not branch into separate inclusion and exclusion algorithms. The two cases differ only in what is occupying the leaf at the resulting index.

\noindent\textit{Inclusion case}: For participant $u$, the leaf at index $\mathcal{N}_u$ holds $\varphi^0_{\mathcal{N}_u} = \mathcal{L}^s(\delta_u,\mu_u,\tau_u)$, the genuine salted leaf transform of $u$'s datum (Property \ref{prop:leaf_identity}). \texttt{LTRVerify} returns $\Phi_{\mathcal{L}^s}^{u}=1$ against this leaf, and each of the $K$ MRP hops along the path verifies via \texttt{MRPHopVerify}, recomputing $H^{curr}$ from the sibling supplied at each level. The final $H^{curr}$ matches $H^{root}$, so \texttt{VerInc} returns \texttt{True} which is indicative of the fact that $u$'s genuine data was used in the analysis, which guarantees that $\mathcal{M}([u,\delta_u]\,|\,R) = 1$.

\noindent\textit{Exclusion case}: For identity $u'$, Property \ref{prop:leaf_identity} gives the leaf hash $H_{u'}$ and index $\mathcal{N}_{u'}$ exactly as it would for a genuine participant, so the path to be checked is encoded from $u'$'s real hash. The leaf at $\mathcal{N}_{u'}$, however, holds the default value $\varphi^0_{\mathcal{N}_{u'}} = \mathcal{L}^s(\varnothing)$. By Property \ref{prop:leaf_unoccupied}, this proves $u'$'s exclusion directly indicating that the path is genuine, but the source at that leaf is the default, not a real LTR witness. The $K$ MRP hops from this default leaf to the root verify exactly as in the inclusion case, and $H^{curr}$ matches $H^{root}$, certifying $\mathcal{M}([u',\delta_{u'}]\,|\,R) = 0$. \autoref{tab:worked_example} summarizes both traces.

\subsection{Case Study Results}
We present the results of our case studies in terms of the witness generation times as well as the computational system performance of the end-to-end LTR and MRP proving mechanism for KS, LR and ACC Tests. We present the proving and verification keys sizes for each case study in Appendix \ref{apx:sup_res}.
\subsubsection{Kolmogorov-Smirnov (KS) Test}: As reported in \autoref{tab:ks_breakdown}, circuit compilation and witness generation times remain tightly bounded across all \texttt{EZKL Scale} configurations. Increasing arithmetic precision introduces only minor fluctuations in total generation time, indicating that higher fixed-point precision does not significantly impact the computational cost of the KS workflow. More importantly, the computed max-gap statistic remains invariant across all precision settings. This invariance demonstrates that increasing zero-knowledge precision does not distort the underlying non-parametric hypothesis test, confirming that \cosmetic preserves numerical correctness under cryptographic constraints.
\begin{table}[!htb]
\setlength{\tabcolsep}{1.5pt}
\centering
\begin{tabular}{|c|cc|cc|c|}
\hline
\rule{0pt}{5ex}
\textbf{\shortstack{EZKL\\Scale}} &
\multicolumn{2}{c|}{\makebox[1cm][c]{\textbf{Model 1}}} &
\multicolumn{2}{c|}{\makebox[1cm][c]{\textbf{Model 2}}} &
\textbf{\shortstack{Total\\Time}}\\
\cline{2-5}
 & \textbf{Circuit} & \textbf{Witness} & \textbf{Circuit} & \textbf{Witness}& \textbf{ } \\
\hline
8  & 307.53 & 1070.87 & 292.40 & 1082.41 & 2916.16 \\
10 & 303.07 & 1063.42 & 306.40 & 1070.08 & 2867.4  \\
12 & 301.02 & 1069.79 & 299.44 & 1078.96 & 2874.8  \\
14 & 291.00 & 1080.10 & 307.74 & 1081.95 & 2892.24 \\
\hline
\end{tabular}
\caption{KS circuit and witness generation time (seconds) across EZKL scales}
\label{tab:ks_breakdown}
\end{table}
\begin{table}[!htb]
\setlength{\tabcolsep}{3pt}
\centering
\begin{tabular}{|c|cc|cc|c|}
\hline
\rule{0pt}{5ex}
\textbf{\shortstack{EZKL\\Scale}} &
\multicolumn{2}{c|}{\textbf{Full Model}} &
\multicolumn{2}{c|}{\textbf{Reduced Model}} &
\textbf{\shortstack{Total\\Time}} \\
\cline{2-5}
 & \textbf{Circuit} & \textbf{Witness} &
   \textbf{Circuit} & \textbf{Witness} & \textbf{ } \\
\hline
8  & 297.63 & 971.44 & 293.57 & 975.65 & 2641.62 \\
10 & 267.61 & 883.93 & 271.59 & 880.93 & 2394.99 \\
12 & 270.30 & 887.07 & 258.56 & 869.07 & 2459.20 \\
14 & 274.54 & 888.47 & 276.65 & 896.43 & 2523.01 \\
\hline
\end{tabular}
\caption{LRT circuit and witness generation time (seconds) across EZKL scales}
\label{tab:lrt_breakdown}
\end{table}

\begin{table}[!htb]
\setlength{\tabcolsep}{2pt}
\centering
\begin{tabular}{|c|cc|cc|c|}
\hline
\rule{0pt}{6ex}
\textbf{\shortstack{EZKL\\Scale}} &
\multicolumn{2}{c|}{\textbf{Length}} &
\multicolumn{2}{c|}{\textbf{Accuracy}} &
\textbf{\shortstack{Total\\Time}} \\
\cline{2-5}
 & \textbf{Circuit} & \textbf{Witness} & \textbf{Circuit} & \textbf{Witness} &  \\
\hline
8  & 343.72 & 1181.16 & 323.31 & 1163.35 & 3011.54 \\
10 & 341.87 & 1171.05 & 292.86 & 1154.64 & 2960.42 \\
12 & 345.77 & 1080.57 & 324.30 & 1122.04 & 2872.68 \\
14 & 347.23 & 1089.28 & 325.69 & 1178.08 & 2940.28 \\
\hline
\end{tabular}
\caption{ACC circuit and witness generation time (seconds) across EZKL scales}
\label{tab:acc_breakdown}
\end{table}


\begin{table}[!htb]
\setlength{\tabcolsep}{4pt}
\centering
\begin{tabular}{|c|ccc|ccc|}
\hline
\rule{0pt}{5ex}
\textbf{\shortstack{EZKL\\Scale}} &
\multicolumn{3}{c|}{\textbf{Max Gap}} &
\multicolumn{3}{c|}{\textbf{LRT Statistic}} \\
\cline{2-7}
 & \textbf{PK} & \textbf{VK} & \textbf{Time} &
   \textbf{PK} & \textbf{VK} & \textbf{Time} \\
\hline
8  & 13.24 & 3.07 & 162.95 & 10.29 & 2.43 & 103.31 \\
10 & 13.24 & 3.07 & 124.43 & 10.29 & 2.43 & 125.65 \\
12 & 13.24 & 3.07 & 125.59 & 10.29 & 2.43 & 174.20 \\
14 & 13.24 & 3.07 & 131.45 & 10.29 & 2.43 & 186.93 \\
\hline
\end{tabular}
\caption{KS and LRT statistics Proving(GB), verification(MB) and generation time (seconds) key sizes across EZKL Scale}
\label{tab:derived_keysize}
\end{table}
\begin{figure*}[!htb]
\centering
\noindent
\subfigure[KS Test]{
  \includegraphics[width=0.31\textwidth,keepaspectratio]{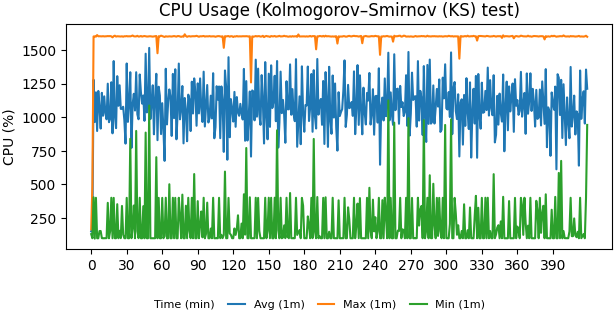}
  \label{fig:cpu_ks}
}
\subfigure[LRT Test]{
  \includegraphics[width=0.31\textwidth,keepaspectratio]{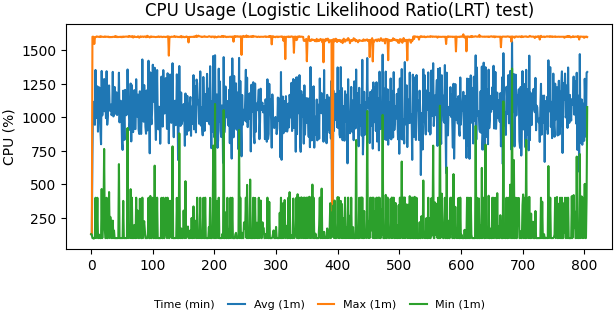}
  \label{fig:cpu_lrt}
}
\subfigure[ACC Test]{
  \includegraphics[width=0.31\textwidth,keepaspectratio]{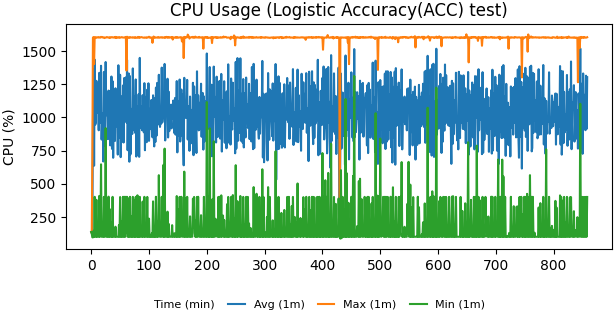}
  \label{fig:cpu_acc}
}
\caption{Prover CPU utilization during KS, LRT, and ACC proof generation.}
\label{fig:cpu_metrics}
\end{figure*}

\begin{figure*}[!htb]
\centering
\noindent
\subfigure[KS Test]{
  \includegraphics[width=0.31\textwidth,keepaspectratio]{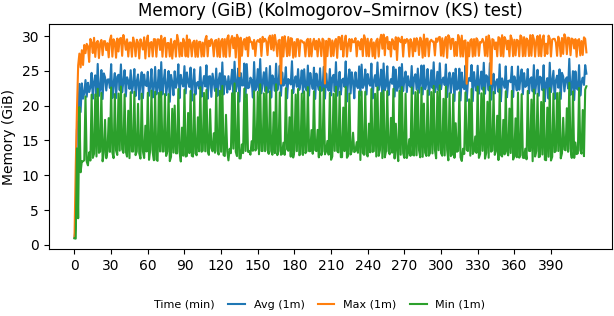}
  \label{fig:memory_ks}
}
\subfigure[LRT Test]{
  \includegraphics[width=0.31\textwidth,keepaspectratio]{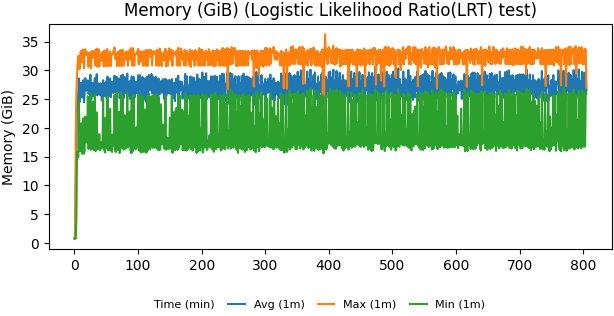}
  \label{fig:memory_lrt}
}
\subfigure[ACC Test]{
  \includegraphics[width=0.31\textwidth,keepaspectratio]{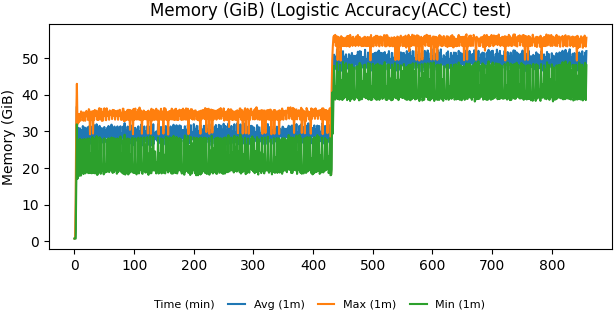}
  \label{fig:memory_acc}
}
\caption{Prover memory utilization during KS, LRT, and ACC proof generation.}
\label{fig:memory_metrics}
\end{figure*}

\subsubsection{Logistic Likelihood-Ratio Test (LRT)}: As shown in \autoref{tab:lrt_breakdown}, both circuit and witness generation times remain stable as the \texttt{EZKL Scale} increases. Despite higher arithmetic precision, the LR statistic is preserved exactly across all configurations, demonstrating that \cosmetic maintains numerical fidelity for parametric hypothesis testing under zero-knowledge constraints. This property is particularly important in regulated settings, where even small numerical deviations can affect downstream clinical or regulatory decisions.

\subsubsection{Logistic Accuracy (ACC) Test}: As reported in \autoref{tab:acc_breakdown}, both transformers incur relatively modest circuit and witness generation times compared to the KS and LRT workflows. Across all \texttt{EZKL Scale} values, the observed classification accuracy remains unchanged, indicating that zero-knowledge proof construction does not interfere with model evaluation outcomes. Moreover, we notice that the total generation times vary only slightly with increasing precision, suggesting predictable performance behavior. This confirms that repeated accuracy verification can be supported efficiently within the \cosmetic framework.

Across all evaluated statistical workflows, \cosmetic demonstrates strong scalability and numerical stability with increasing zero-knowledge precision. With respect to statistical stability, across all zkSNARK precision scales ranging from 8-14 setting, all three statistical outputs remain identical with the KS Max Gap statistic is 0.999, the LRT statistic is 386.78, and the ACC accuracy is 0.749. This invariance shows that \cosmetic's quantization introduces no numerical drift regardless of precision level. Additionally, our results also demonstrate stability in circuit and witness generation times across diverse scale values as well while preserving exact statistical outputs. These results collectively confirm that \cosmetic enables privacy-preserving statistical inference without compromising numerical correctness or cryptographic scalability.


\noindent\textbf{System-level CPU utilization}: Figures~\ref{fig:cpu_ks}--\ref{fig:cpu_acc} show the prover CPU utilization during KS, LRT, and ACC proof generation, respectively. For the KS test (Figure~\ref{fig:cpu_ks}), CPU usage exhibits periodic fluctuations, reflecting alternating computation phases associated with histogram aggregation and cumulative distribution comparison within the KS circuit. Despite these fluctuations, overall CPU utilization remains consistently elevated, indicating sustained computational activity throughout the proving process.

The LRT workload (Figure~\ref{fig:cpu_lrt}) demonstrates a more uniform CPU utilization profile, corresponding to continuous arithmetic computation required to evaluate both the full and reduced logistic regression models within the zero-knowledge circuit. The absence of pronounced oscillations suggests a steady, compute-bound execution pattern.

Similarly, the ACC workflow (Figure~\ref{fig:cpu_acc}) shows sustained CPU utilization across the entire execution window, indicating continuous computation during both the length SMT proof phase and the subsequent accuracy proof phase. The lack of extended idle periods suggests that ACC proof generation maintains a stable computational workload across phases.

\noindent\textbf{System-level memory utilization.}
Figures~\ref{fig:memory_ks}--\ref{fig:memory_acc} present the prover memory utilization during KS, LRT, and ACC proof generation. For the KS test (Figure~\ref{fig:memory_ks}), memory usage remains stable throughout execution, with minimal divergence between minimum, average, and maximum values. This behavior indicates that KS proof generation relies on a fixed working set after initialization and does not require repeated large memory allocations.

In the LRT case (Figure~\ref{fig:memory_lrt}), memory utilization also exhibits a steady profile with limited variation over time. The absence of large memory spikes suggests efficient reuse of intermediate buffers during witness generation and likelihood computation.

For the ACC workflow (Figure~\ref{fig:memory_acc}), memory usage remains stable during the initial execution period and increases when the system transitions from length SMT proof generation to accuracy proof generation. After this transition, memory utilization stabilizes at a higher level, indicating predictable, phase-dependent memory usage rather than unbounded allocation
growth.

\section{Conclusion}
In this paper, we introduce \cosmetic, a zero-knowledge framework based on Computational Sparse Merkle Trees (CSMTs) that enables verifiable statistical computation while preserving patient privacy in clinical research. By leveraging Merkle tree structure and computational reduction operations, CSMTs support succinct inclusion and exclusion proofs for individual users whose data participate in specific stages of a statistical analysis. \cosmetic integrates zk-SNARKs to provide end-to-end, publicly verifiable guarantees of data membership and exclusion across the computational pipeline. We formalize soundness conditions for these proofs using state-of-the-art proving systems such as Halo2 and demonstrate that \cosmetic can also verify claims of exclusive data usage from subsets of large medical databases. Through diverse real-world case studies, we show that the framework is extensible to a wide range of statistical methods used in clinical research.

 Unlike prior zero-knowledge Merkle tree systems, which treat leaves as static key-value records and certify only lookup or boolean membership, and unlike verifiable-statistics frameworks, which certify a global test statistic without binding it to individual participants, \cosmetic is the first to combine analysis-specific computational reductions at internal nodes with participant-level inclusion and exclusion proofs for the same statistical pipeline. 
To evaluate \cosmetic, we design three representative use cases using real-world clinical data from Huntington’s disease (HD) and HIV-1 studies. First, we implement a two-sample Kolmogorov–Smirnov (KS) test to verify whether CAG repeat length distributions differ between HD patients and healthy controls, while proving that only approved cohort data were used and all non-participants were excluded. Second, we construct a publicly verifiable likelihood ratio test (LRT) for HIV-1 resistance to Temsavir using nested logistic-regression models to assess the significance of selected genomic features. Third, we extend the framework to support verifiable evaluation of predictive accuracy under the same inclusion and exclusion guarantees. Across multiple zero-knowledge scales, we demonstrate that circuit and witness generation times, statistical outputs, and key sizes remain practical and stable, showing that privacy-preserving hypothesis testing can be achieved without sacrificing computational efficiency or statistical validity.

Finally, we show that \cosmetic enables regulators to verify the authenticity and correct use of datasets underlying clinical studies, while allowing participants to obtain publicly verifiable proofs of data inclusion or exclusion. By combining correctness, privacy, and transparency, \cosmetic offers a practical and scalable foundation for trustworthy, auditable, and privacy-preserving clinical research.

\begin{acks}
\noindent \textit{Funding Support}: This material is based upon work supported by the National Science Foundation (NSF) under Grant No. 2348411 and by National Institutes of Health (NIH) grant R01 AI170205.
\end{acks}
\bibliographystyle{ACM-Reference-Format} 
\bibliography{main}

@inproceedings{len2023elektra,
  title={ELEKTRA: Efficient lightweight multi-dEvice key TRAnsparency},
  author={Len, Julia and Chase, Melissa and Ghosh, Esha and Jost, Daniel and Kesavan, Balachandar and Marcedone, Antonio},
  booktitle={Proceedings of the 2023 ACM SIGSAC Conference on Computer and Communications Security},
  pages={2915--2929},
  year={2023}
}

@article{malvai2023parakeet,
  title={Parakeet: Practical key transparency for end-to-end encrypted messaging},
  author={Malvai, Harjasleen and Kokoris-Kogias, Lefteris and Sonnino, Alberto and Ghosh, Esha and Ozt{\"u}rk, Ercan and Lewi, Kevin and Lawlor, Sean},
  journal={Cryptology ePrint Archive},
  year={2023}
}

@inproceedings{wang2023balanceproofs,
  title={$\{$BalanceProofs$\}$: Maintainable vector commitments with fast aggregation},
  author={Wang, Weijie and Ulichney, Annie and Papamanthou, Charalampos},
  booktitle={32nd USENIX Security Symposium (USENIX Security 23)},
  pages={4409--4426},
  year={2023}
}

@article{liu2021evaluating,
  title={Evaluating eligibility criteria of oncology trials using real-world data and AI},
  author={Liu, Ruishan and Rizzo, Shemra and Whipple, Samuel and Pal, Navdeep and Pineda, Arturo Lopez and Lu, Michael and Arnieri, Brandon and Lu, Ying and Capra, William and Copping, Ryan and others},
  journal={Nature},
  volume={592},
  number={7855},
  pages={629--633},
  year={2021},
  publisher={Nature Publishing Group UK London}
}

@misc{ezkl2024,
  title        = {ezkl: Zero-Knowledge Machine Learning Inference Framework},
  author       = {South, Tobin and Camuto, Alexander and contributors},
  howpublished = {\url{https://github.com/zkonduit/ezkl}},
  note         = {Accessed: 2025-02-12},
  year         = {2024}
}

@inproceedings{food2024evaluating,
  title={Evaluating inclusion and exclusion criteria in clinical trials},
  author={Food and Drug Administration and others},
  booktitle={Workshop Report.[online]. Accessed},
  volume={13},
  year={2024}
}

@article{dean2008mapreduce,
  title={MapReduce: simplified data processing on large clusters},
  author={Dean, Jeffrey and Ghemawat, Sanjay},
  journal={Communications of the ACM},
  volume={51},
  number={1},
  pages={107--113},
  year={2008},
  publisher={ACM New York, NY, USA}
}

@article{dongarra1995introduction,
  title={An introduction to the MPI standard},
  author={Dongarra, Jack J and Otto, Steve W and Snir, Marc and Walker, David and others},
  journal={Communications of the ACM},
  volume={18},
  number={11},
  year={1995}
}

@inproceedings{zaharia2010spark,
  title={Spark: Cluster computing with working sets},
  author={Zaharia, Matei and Chowdhury, Mosharaf and Franklin, Michael J and Shenker, Scott and Stoica, Ion},
  booktitle={2nd USENIX workshop on hot topics in cloud computing (HotCloud 10)},
  year={2010}
}

@article{tzialla2021transparency,
  author    = {Ioanna Tzialla and Abhiram Kothapalli and Bryan Parno and Srinath Setty},
  title     = {Transparency Dictionaries with Succinct Proofs of Correct Operation},
  journal   = {{IACR} Cryptology ePrint Archive},
  volume    = {2021},
  pages     = {1263},
  year      = {2021},
  url       = {https://eprint.iacr.org/2021/1263}
}

@article{kurbatov2024imok,
  author       = {Oleksandr Kurbatov and
                  Lasha Antadze and
                  Ameen Soleimani and
                  Kyrylo Riabov and
                  Artem Sdobnov},
  title        = {{IMOK:} A Compact Connector for Non-prohibition Proofs to
                  Privacy-Preserving Applications},
  journal      = {{IACR} Cryptology ePrint Archive},
  volume       = {2024},
  pages        = {1868},
  year         = {2024},
  url          = {https://eprint.iacr.org/2024/1868}
}

@misc{chystiakov2025cmt,
  author       = {Artem Chystiakov and Oleh Komendant and Kyrylo Riabov},
  title        = {Cartesian Merkle Tree},
  howpublished = {arXiv preprint arXiv:2504.10944},
  year         = {2025},
  url          = {https://arxiv.org/abs/2504.10944}
}

@article{rizzini2025private,
  title={A Private Smart Wallet with Probabilistic Compliance},
  author={Rizzini, Andrea and Esposito, Marco and Bruschi, Francesco and Sciuto, Donatella},
  journal={arXiv preprint arXiv:2506.04853},
  year={2025}
}

@misc{snarkedsummtree,
  author       = {Brandon Gomes},
  title        = {{SNARKed Merkle Sum Tree:} A Practical Proof-of-Solvency      Protocol based on Vitalik's Proposal},
  howpublished = {Ethereum Research forum post},
  year         = {2022},
  url          = {https://ethresear.ch/t/snarked-merkle-sum-tree-a-practical-proof-of-solvency-protocol-based-on-vitaliks-proposal/14405},
  note         = {Accessed November 2025}
}

@phdthesis{servan2018cryptographically,
  title={Cryptographically Certified Hypothesis Testing},
  author={Servan-Schreiber, Sacha},
  year={2018},
  school={Brown University}
}

@inproceedings{narayan2015verdp,
  author    = {Arjun Narayan and Ariel J. Feldman and
               Antonis Papadimitriou and Andreas Haeberlen},
  title     = {Verifiable Differential Privacy},
  booktitle = {Proceedings of the 10th European Conference on Computer Systems (EuroSys)},
  year      = {2015},
  publisher = {ACM}
}

@article{zhu2024risefl,
  author    = {Yizheng Zhu and Yuncheng Wu and Zhaojing Luo and
               Beng Chin Ooi and Xiaokui Xiao},
  title     = {Secure and Verifiable Data Collaboration with Low-Cost Zero-Knowledge Proofs},
  journal   = {Proceedings of the VLDB Endowment},
  volume    = {17},
  number    = {9},
  pages     = {2321--2334},
  year      = {2024},
  doi       = {10.14778/3665844.3665860},
  url       = {https://www.vldb.org/pvldb/vol17/p2321-zhu.pdf}
}

@mastersthesis{lonfils2023electionguard,
  author = {Marie Lonfils},
  title  = {Risk-Limiting Audit Optimization with ElectionGuard:
            Zero-Knowledge Arguments on Chaum{-}Pedersen Multi-Commitments},
  school = {Universit{\'e} catholique de Louvain},
  year   = {2023},
  address= {Louvain-la-Neuve, Belgium},
  url    = {https://thesis.dial.uclouvain.be/entities/masterthesis/5143cd5c-4284-49e5-9225-72a5f3b9ceb0},
  note   = {Master's thesis}
}

@article{MengZhang2024,
    author = {Meng Zhang  and Yongqi Zheng  and Xiagela Maidaiti  and Baosheng Liang  and Yongyue Wei  and Feng Sun },
    title = {Integrating Machine Learning into Statistical Methods in Disease Risk Prediction Modeling: A Systematic Review},
    journal = {Health Data Science},
    volume = {4},
    number = {},
    pages = {0165},
    year = {2024},
    doi = {10.34133/hds.0165},
    URL = {https://spj.science.org/doi/abs/10.34133/hds.0165},
    eprint = {https://spj.science.org/doi/pdf/10.34133/hds.0165}
}

@article{Bi2022,
  title     = "{COVID-19} forecasting and intervention planning using gated
               recurrent unit and evolutionary algorithm",
  author    = "Bi, Luning and Fili, Mohammad and Hu, Guiping",
  journal   = "Neural Comput. Appl.",
  publisher = "Springer Science and Business Media LLC",
  volume    =  34,
  number    =  20,
  pages     = "17561--17579",
  month     =  may,
  year      =  2022
}

@ARTICLE{Guerra2025,
  title     = "Towards pharmacological prevention of Alzheimer disease",
  author    = "Llibre-Guerra, Jorge J and McDade, Eric M and Schindler, Suzanne
               E and Clifford, David B and Supnet, Charlene and Atri, Alireza
               and Bateman, Randall J",
  journal   = "Nat. Rev. Neurol.",
  publisher = "Springer Science and Business Media LLC",
  volume    =  21,
  number    =  12,
  pages     = "721--733",
  month     =  dec,
  year      =  2025
}

@ARTICLE{Huang2018,
  title     = "Machine learning predicts individual cancer patient responses to
               therapeutic drugs with high accuracy",
  author    = "Huang, Cai and Clayton, Evan A and Matyunina, Lilya V and
               McDonald, L Deette and Benigno, Benedict B and Vannberg, Fredrik
               and McDonald, John F",
  journal   = "Sci. Rep.",
  publisher = "Springer Science and Business Media LLC",
  volume    =  8,
  number    =  1,
  pages     = "16444",
  month     =  nov,
  year      =  2018,
  copyright = "https://creativecommons.org/licenses/by/4.0",
  language  = "en"
}

@ARTICLE{Xie2017,
  title     = "Discovery of novel therapeutic properties of drugs from
               transcriptional responses based on multi-label classification",
  author    = "Xie, Lingwei and He, Song and Wen, Yuqi and Bo, Xiaochen and
               Zhang, Zhongnan",
  journal   = "Sci. Rep.",
  publisher = "Springer Science and Business Media LLC",
  volume    =  7,
  number    =  1,
  pages     = "7136",
  month     =  aug,
  year      =  2017,
  copyright = "https://creativecommons.org/licenses/by/4.0",
  language  = "en"
}

@article{Rojas2024,
    doi = {10.1371/journal.pcbi.1012215},
    author = {Rojas Chávez, Roberth Anthony AND Fili, Mohammad AND Han, Changze AND Rahman, Syed A. AND Bicar, Isaiah G. L. AND Gregory, Sullivan AND Helverson, Annika AND Hu, Guiping AND Darbro, Benjamin W. AND Das, Jishnu AND Brown, Grant D. AND Haim, Hillel},
    journal = {PLOS Computational Biology},
    publisher = {Public Library of Science},
    title = {Mapping the Evolutionary Space of SARS-CoV-2 Variants to Anticipate Emergence of Subvariants Resistant to COVID-19 Therapeutics},
    year = {2024},
    month = {06},
    volume = {20},
    url = {https://doi.org/10.1371/journal.pcbi.1012215},
    pages = {1-33},
    number = {6},
}

@ARTICLE{Shih2017,
  title     = "A five-region hypothesis test for exposure-disease associations",
  author    = "Shih, Han-Yi and Lee, Wen-Chung",
  journal   = "Sci. Rep.",
  publisher = "Springer Science and Business Media LLC",
  volume    =  7,
  number    =  1,
  pages     = "5131",
  month     =  jul,
  year      =  2017,
  copyright = "https://creativecommons.org/licenses/by/4.0",
  language  = "en"
}

@ARTICLE{Qiu2025,
  title     = "{MetaboLM}: a metabolomic language model for multi-disease early
               prediction and risk stratification",
  author    = "Qiu, Shizheng and Guo, Jirui and Zhang, Zhishuai and Liang,
               Haozheng and You, Huanyu and Hu, Yang and Liu, Guiyou and Wang,
               Yadong",
  journal   = "Nat. Commun.",
  publisher = "Springer Science and Business Media LLC",
  volume    =  16,
  number    =  1,
  pages     = "11272",
  month     =  dec,
  year      =  2025,
  copyright = "https://creativecommons.org/licenses/by-nc-nd/4.0",
  language  = "en"
}

@ARTICLE{Ward2020,
  title     = "Machine learning and atherosclerotic cardiovascular disease risk
               prediction in a multi-ethnic population",
  author    = "Ward, Andrew and Sarraju, Ashish and Chung, Sukyung and Li,
               Jiang and Harrington, Robert and Heidenreich, Paul and
               Palaniappan, Latha and Scheinker, David and Rodriguez, Fatima",
  journal   = "NPJ Digit. Med.",
  publisher = "Springer Science and Business Media LLC",
  volume    =  3,
  number    =  1,
  pages     = "125",
  month     =  sep,
  year      =  2020,
  copyright = "https://creativecommons.org/licenses/by/4.0",
  language  = "en"
}

@ARTICLE{Billah2025,
  title     = "Enabling {FAIR} data stewardship in complex international
               multi-site studies: Data Operations for the Accelerating
               Medicines Partnership\textregistered{} Schizophrenia Program",
  author    = "Billah, Tashrif and Cho, Kang Ik K and Borders, Owen and Chung,
               Yoonho and Ennis, Michaela and Jacobs, Grace R and Liebenthal,
               Einat and Mathalon, Daniel H and Mohandass, Dheshan and
               Nicholas, Spero C and Pasternak, Ofer and Penzel, Nora and
               Eichi, Habiballah Rahimi and Wolff, Phillip and Anticevic, Alan
               and Laulette, Kristen and Nunez, Angela R and Tamayo, Zailyn and
               Buccilli, Kate and Colton, Beau-Luke and Dwyer, Dominic B and
               Hendricks, Larry and Yuen, Hok Pan and Spark, Jessica and Tod,
               Sophie and Carrington, Holly and Chen, Justine T and Coleman,
               Michael J and Corcoran, Cheryl M and Haidar, Anastasia and John,
               Omar and Kelly, Sinead and Marcy, Patricia J and Matneja, Priya
               and McGowan, Alessia and Ray, Susan E and Veale, Simone and
               Winter-Van Rossum, Inge and Addington, Jean and Allott, Kelly A
               and Calkins, Monica E and Clark, Scott R and Gur, Ruben C and
               Harms, Michael P and Perkins, Diana O and Ruparel, Kosha and
               Stone, William S and Torous, John and Yung, Alison R and Zoupou,
               Eirini and Fusar-Poli, Paolo and Mittal, Vijay A and Shah, Jai L
               and Wolf, Daniel H and Cecchi, Guillermo and Kapur, Tina and
               Kubicki, Marek and Lewandowski, Kathryn Eve and Bearden, Carrie
               E and McGorry, Patrick D and Kahn, Ren{\'e} S and Kane, John M
               and Nelson, Barnaby and Woods, Scott W and Shenton, Martha E and
               {Accelerating Medicines Partnership\textregistered{}
               Schizophrenia (AMP\textregistered{} SCZ)} and Baker, Justin T
               and Bouix, Sylvain",
  journal   = "Schizophrenia (Heidelb.)",
  publisher = "Springer Science and Business Media LLC",
  volume    =  11,
  number    =  1,
  pages     = "55",
  month     =  apr,
  year      =  2025,
  copyright = "https://creativecommons.org/licenses/by/4.0",
  language  = "en"
}

@ARTICLE{Kelly2019,
  title     = "Key challenges for delivering clinical impact with artificial
               intelligence",
  author    = "Kelly, Christopher J and Karthikesalingam, Alan and Suleyman,
               Mustafa and Corrado, Greg and King, Dominic",
  journal   = "BMC Med.",
  publisher = "Springer Science and Business Media LLC",
  volume    =  17,
  number    =  1,
  pages     = "195",
  month     =  oct,
  year      =  2019,
  copyright = "http://creativecommons.org/licenses/by/4.0/",
  language  = "en"
}

@ARTICLE{Curzon2021,
  author={Curzon, James and Kosa, Tracy Ann and Akalu, Rajen and El-Khatib, Khalil},
  journal={IEEE Transactions on Artificial Intelligence}, 
  title={Privacy and Artificial Intelligence}, 
  year={2021},
  volume={2},
  number={2},
  pages={96-108},
  doi={10.1109/TAI.2021.3088084}}

@ARTICLE{Falco2021,
  title     = "Governing {AI} safety through independent audits",
  author    = "Falco, Gregory and Shneiderman, Ben and Badger, Julia and
               Carrier, Ryan and Dahbura, Anton and Danks, David and Eling,
               Martin and Goodloe, Alwyn and Gupta, Jerry and Hart, Christopher
               and Jirotka, Marina and Johnson, Henric and LaPointe, Cara and
               Llorens, Ashley J and Mackworth, Alan K and Maple, Carsten and
               P{\'a}lsson, Sigur{\dh}ur Emil and Pasquale, Frank and Winfield,
               Alan and Yeong, Zee Kin",
  journal   = "Nat. Mach. Intell.",
  publisher = "Springer Science and Business Media LLC",
  volume    =  3,
  number    =  7,
  pages     = "566--571",
  month     =  jul,
  year      =  2021,
  copyright = "https://www.springernature.com/gp/researchers/text-and-data-mining",
  language  = "en"
}

@ARTICLE{Malin2010,
  title     = "Technical and policy approaches to balancing patient privacy and
               data sharing in clinical and translational research",
  author    = "Malin, Bradley and Karp, David and Scheuermann, Richard H",
  journal   = "J. Investig. Med.",
  publisher = "BMJ",
  volume    =  58,
  number    =  1,
  pages     = "11--18",
  month     =  jan,
  year      =  2010,
  language  = "en"
}

@misc{halo2_zcash_2023,
  title        = {Halo2: The Halo2 Zero-Knowledge Proving System},
  author       = {{Zcash Team}},
  year         = {2023},
  howpublished = {\url{https://zcash.github.io/halo2/concepts/proofs.html}},
  note         = {Accessed: 2026-01-13}
}

@ARTICLE{Markham2020,
  title     = "Fostemsavir: First approval",
  author    = "Markham, Anthony",
  journal   = "Drugs",
  publisher = "Springer Science and Business Media LLC",
  volume    =  80,
  number    =  14,
  pages     = "1485--1490",
  month     =  sep,
  year      =  2020,
  language  = "en"
}

@ARTICLE{Wang2018_xk,
  title     = "Discovery of the human immunodeficiency virus type 1 ({HIV-1})
               attachment inhibitor temsavir and its phosphonooxymethyl prodrug
               fostemsavir",
  author    = "Wang, Tao and Ueda, Yasu and Zhang, Zhongxing and Yin, Zhiwei
               and Matiskella, John and Pearce, Bradley C and Yang, Zheng and
               Zheng, Ming and Parker, Dawn D and Yamanaka, Gregory A and Gong,
               Yi-Fei and Ho, Hsu-Tso and Colonno, Richard J and Langley, David
               R and Lin, Pin-Fang and Meanwell, Nicholas A and Kadow, John F",
  journal   = "J. Med. Chem.",
  publisher = "American Chemical Society (ACS)",
  volume    =  61,
  number    =  14,
  pages     = "6308--6327",
  month     =  jul,
  year      =  2018,
  language  = "en"
}

@ARTICLE{Macdonald1993_mp,
  title     = "A novel gene containing a trinucleotide repeat that is expanded
               and unstable on Huntington's disease chromosomes",
  author    = "Macdonald, M",
  journal   = "Cell",
  publisher = "Elsevier BV",
  volume    =  72,
  number    =  6,
  pages     = "971--983",
  month     =  mar,
  year      =  1993,
  language  = "en"
}

@ARTICLE{Lee2012_hn,
  title     = "{CAG} repeat expansion in Huntington disease determines age at
               onset in a fully dominant fashion",
  author    = "Lee, J-M and Ramos, E M and Lee, J-H and Gillis, T and Mysore, J
               S and Hayden, M R and Warby, S C and Morrison, P and Nance, M
               and Ross, C A and Margolis, R L and Squitieri, F and Orobello, S
               and Di Donato, S and Gomez-Tortosa, E and Ayuso, C and
               Suchowersky, O and Trent, R J A and McCusker, E and Novelletto,
               A and Frontali, M and Jones, R and Ashizawa, T and Frank, S and
               Saint-Hilaire, M H and Hersch, S M and Rosas, H D and Lucente, D
               and Harrison, M B and Zanko, A and Abramson, R K and Marder, K
               and Sequeiros, J and Paulsen, J S and {PREDICT-HD study of the
               Huntington Study Group (HSG)} and Landwehrmeyer, G B and
               {REGISTRY study of the European Huntington's Disease Network}
               and Myers, R H and {HD-MAPS Study Group} and MacDonald, M E and
               Gusella, J F and {COHORT study of the HSG}",
  journal   = "Neurology",
  publisher = "Ovid Technologies (Wolters Kluwer Health)",
  volume    =  78,
  number    =  10,
  pages     = "690--695",
  month     =  mar,
  year      =  2012,
  language  = "en"
}

@ARTICLE{Moily2014_wn,
  title     = "Trinucleotide repeats and haplotypes at the huntingtin locus in
               an Indian sample overlaps with European haplogroup a",
  author    = "Moily, Nagaraj S and Kota, Lakshmi Narayanan and Anjanappa, Ram
               Murthy and Venugopal, Sowmya and Vaidyanathan, Radhika and Pal,
               Pramod and Purushottam, Meera and Jain, Sanjeev and Kandasamy,
               Mahesh",
  journal   = "PLoS Curr.",
  publisher = "Public Library of Science (PLoS)",
  volume    =  6,
  month     =  sep,
  year      =  2014,
  language  = "en"
}

@ARTICLE{Jiang2014_we,
  title     = "Huntingtin gene {CAG} repeat numbers in Chinese patients with
               Huntington's disease and controls",
  author    = "Jiang, H and Sun, Y M and Hao, Y and Yan, Y P and Chen, K and
               Xin, S H and Tang, Y P and Li, X H and Jun, T and Chen, Y Y and
               Liu, Z J and Wang, C R and Li, H and Pei, Z and Shang, H F and
               Zhang, B R and Gu, W H and Wu, Z Y and Tang, B S and Burgunder,
               J-M and {Chinese HD Network}",
  journal   = "Eur. J. Neurol.",
  publisher = "Wiley",
  volume    =  21,
  number    =  4,
  pages     = "637--642",
  month     =  apr,
  year      =  2014,
  copyright = "http://onlinelibrary.wiley.com/termsAndConditions\#vor",
  language  = "en"
}

@ARTICLE{Gardiner2017_vg,
  title    = "Huntingtin gene repeat size variations affect risk of lifetime
              depression",
  author   = "Gardiner, Sarah L and van Belzen, Martine J and Boogaard, Merel W
              and van Roon-Mom, Willeke M C and Rozing, Maarten P and van
              Hemert, Albert M and Smit, Johannes H and Beekman, Aartjan T F
              and van Grootheest, Gerard and Schoevers, Robert A and Oude
              Voshaar, Richard C and Roos, Raymund A C and Comijs, Hannie C and
              Penninx, Brenda W J H and van der Mast, Roos C and Aziz, N Ahmad",
  journal  = "Transl. Psychiatry",
  volume   =  7,
  number   =  12,
  pages    = "1277",
  month    =  dec,
  year     =  2017,
  language = "en"
}

@ARTICLE{Vater2025_fq,
  title    = "Huntingtin {CAG} repeat size variations below the Huntington's
              disease threshold: associations with depression, anxiety and
              basal ganglia structure",
  author   = "Vater, Magdalena and Rost, Nicolas and Eckstein, Gertrud and
              Sauer, Susann and Tontsch, Alina and Erhardt, Angelika and Lucae,
              Susanne and Br{\"u}ckl, Tanja and Klopstock, Thomas and
              S{\"a}mann, Philipp G and Binder, Elisabeth B",
  journal  = "Eur. J. Hum. Genet.",
  volume   =  33,
  number   =  5,
  pages    = "624--632",
  month    =  may,
  year     =  2025,
  language = "en"
}

@ARTICLE{Gartland2021_mz,
  title     = "Prevalence of gp160 polymorphisms known to be related to
               decreased susceptibility to temsavir in different subtypes of
               {HIV-1} in the Los Alamos National Laboratory {HIV} Sequence
               Database",
  author    = "Gartland, Margaret and Arnoult, Eric and Foley, Brian T and
               Lataillade, Max and Ackerman, Peter and Llamoso, Cyril and
               Krystal, Mark",
  journal   = "J. Antimicrob. Chemother.",
  publisher = "Oxford University Press (OUP)",
  volume    =  76,
  number    =  11,
  pages     = "2958--2964",
  month     =  oct,
  year      =  2021,
  copyright = "https://creativecommons.org/licenses/by-nc/4.0/",
  language  = "en"
}

@ARTICLE{Zuze2023_wh,
  title    = "Fostemsavir resistance-associated polymorphisms in {HIV-1}
              subtype {C} in a large cohort of treatment-na{\"\i}ve and
              treatment-experienced individuals in Botswana",
  author   = "Zuze, Boitumelo J L and Radibe, Botshelo T and Choga, Wonderful T
              and Bareng, Ontlametse T and Moraka, Natasha O and Maruapula,
              Dorcas and Seru, Kedumetse and Mokgethi, Patrick and Mokaleng,
              Baitshepi and Ndlovu, Nokuthula and Kelentse, Nametso and
              Pretorius-Holme, Molly and Shapiro, Roger and Lockman, Shahin and
              Makhema, Joseph and Novitsky, Vlad and Seatla, Kaelo K and Moyo,
              Sikhulile and Gaseitsiwe, Simani",
  journal  = "Microbiol. Spectr.",
  volume   =  11,
  number   =  6,
  pages    = "e0125123",
  month    =  dec,
  year     =  2023,
  language = "en"
}

\appendix

\section{Proofs}
\subsection{Proof of Proposition \ref{prs:csmt}}\label{subsec:proof_csmt}
\noindent Case 1.1 (Inclusion implies Membership): For this case it is sufficient to prove that $(u\in U)\implies \mathcal{M}([u,\delta]|R) =1$. If we consider $u\in U$, Properties \ref{prop:user_identity} and \ref{prop:leaf_identity} imply the existence of unique user and transform salts given by $\mu,\tau$ respectively. Leveraging Properties \ref{prop:leaf_identity} and \ref{prop:leaf_unoccupied}, we know that the data tuple $(\delta,\mu,\tau)$ for user $u\in U$ corresponds to a unique leaf index given by $\mathcal{N}_u$ such that the binary representation of $\mathcal{N}_u$ is equal to $\verb|hash|(\mathcal{L}^s(\delta,\mu,\tau))$. Given the collision resistance property of $\verb|hash|$, uniqueness of $\mu$ and $\tau$, we can state that $\mathcal{L}({\delta}) \in \Delta^{\mathcal{L}}_U$ which corresponds to $\mathcal{M}([u,\delta]|R) =1$

\noindent Case 1.2 (Exclusion implies Non Membership): For this case, it is sufficient to prove that $(u\notin U)\implies \mathcal{M}([u,\delta]|R) =0$. If $u\notin U$, then $\mathcal{L}(\delta) \notin \Delta^{\mathcal{L}}_U$ which would mean that $\mathcal{M}([u,\delta]|R) =0$ according to Definition \ref{def_mf}.

\noindent Case 2.1 (Membership implies Inclusion): For this case, it is sufficient to show that $\mathcal{M}([u,\delta]|R) =1 \implies (u\in U)$. We know using Definition \ref{def_mf} that $\mathcal{M}([u,\delta]|R) =1$, can only occur when $\exists \mathcal{L}(\delta) \in \Delta^{\mathcal{L}}_U$. As a result, we can find unique user and transform salt vectors $\mu\in \mathbb{R}^{s_u}$ and $\tau\in\mathbb{R}^{s_t}$ to obtain an augmented data tuple $(\delta,\mu,\tau)$ such that $H = \verb|hash|(\mathcal{L}^s(\delta,\mu,\tau))$. Using Properties \ref{prop:leaf_identity} and \ref{prop:leaf_unoccupied}, we can state that $\exists \mathcal{N} = \verb|Decimal|(H)$ and $\varphi^0_\mathcal{N} \neq \mathcal{L}^s(\varnothing)$ thereby demonstrating inclusion $u\in U$.

\noindent Case 2.2 (Non Membership implies Exclusion): For this case, it is sufficient to demonstrate that $\mathcal{M}([u,\delta]|R) =0 \implies (u\notin U)$. We know that $\mathcal{M}([u,\delta]|R) =0$ occurs only when $\mathcal{L}({\delta}) \notin \Delta^{\mathcal{L}}_U$. Assuming sufficient entropy in the salt strings and given collision resistance, we can say that for all combinations of $\hat\mu\in \mathbb{R}^{s_u},\hat\tau\in\mathbb{R}^{s_t}$ the following relation holds where $H_u = \verb|hash|(\mathcal{L}^s(\delta_u,{\mu}_u,{\tau}_u))$
\begin{equation}
    Pr\Big[ \verb|hash|(\mathcal{L}^s(\delta,\hat{\mu},\hat{\tau})) = H_u\Big] \leq \negl(K) \quad \forall u\in U
\end{equation}
Therefore, it is not possible to find some $\hat{\mu},\hat{\tau}$ for which $\varphi^0_\mathcal{N}\neq \mathcal{L}^s(\varnothing)$ where $\mathcal{N} = \verb|Decimal|(\verb|hash|(\mathcal{L}^s(\delta,\hat{\mu},\hat{\tau})))$

\subsection{Proof of Proposition \ref{prs:pi_mem}}\label{subsec:proof_pi_mem}

\begin{proof}
We will present two cases for completing the proof.\\
\emph{Case 1}: ($\Phi_{\mathcal{L}^s}^u = 1 \text{ and }\ \Phi_{\mathcal{A}^{l}}^{u,k} = 1 \implies \mathcal{M}([u,\delta]|\mathcal{R}(\mathcal{A},\mathcal{L},U))$)\\
For this case, we know that leaf and aggregation proof artifacts are verifiable. 

\noindent Case 1.1: Let us assume that there exists a user $u\in U$ with datum $\delta$ such that 
\begin{gather}
    \Phi_{\mathcal{L}^s}^u = 1 \text{ and }\ \Phi_{\mathcal{A}^{l}}^{u,k} = 1 \label{eq:proofeq11}\\
    \mathcal{M}([u,\delta]|\mathcal{R}(\mathcal{A},\mathcal{L},U))=0
\end{gather}

If $\Phi_{\mathcal{L}^s}^u = 1$ and $\Phi_{\mathcal{A}^{l}}^{u,k} = 1$, then it also implies that there exists some user $u$ such that Definition \ref{defn3} applies. Using Proposition \ref{prs:csmt}, we already know that when $u\in U$ and $\mathcal{M}([u,\delta]|\mathcal{R}(\mathcal{A},\mathcal{L},U))=1$. As a result we reach a contradiction of our earlier stated assumption. 

\noindent Case 1.2: Let us assume that there exists a user $u\notin U$ with datum $\delta$ such that 
\begin{gather}
    \Phi_{\mathcal{L}^s}^u = 1 \text{ and }\ \Phi_{\mathcal{A}^{l}}^{u,k} = 1 \label{eq:proofeq12}\\
    \mathcal{M}([u,\delta]|\mathcal{R}(\mathcal{A},\mathcal{L},U))=1
\end{gather}
Using a similar logic as above, we argue that using Proposition \ref{prs:csmt}, $\mathcal{M}([u,\delta]|\mathcal{R}(\mathcal{A},\mathcal{L},U))=0$ when $u\notin U$ leading us to a contradiction of $\mathcal{M}([u,\delta]|\mathcal{R}(\mathcal{A},\mathcal{L},U))=1$ 

Using both Case 1.1 and Case 1.2, we can clearly state that a user who satisfies $\Phi_{\mathcal{L}^s}^u = 1 \text{ and }\ \Phi_{\mathcal{A}^{l}}^{u,k} = 1$ also necessarily leads to the correct output of the membership function $\mathcal{M}([u,\delta]|\mathcal{R}(\mathcal{A},\mathcal{L},U))$.

\noindent
\emph{Case 2} ($\mathcal{M}([u,\delta]|\mathcal{R}(\mathcal{A},\mathcal{L},U)) \implies \Phi_{\mathcal{L}^s}^u = 1 \text{ and }\ \Phi_{\mathcal{A}^{l}}^{u,k} = 1$)

\noindent To prove sufficiency, let us consider two sub cases.

\noindent Case 2.1: Assume we have a user $u\in U$ with datum $\delta$ such that $\mathcal{M}([u,\delta]|\mathcal{R}(\mathcal{A},\mathcal{L},U))=1$. Let us assume either $\Phi_{\mathcal{L}^s}^u = 0 \text{ or }\ \Phi_{\mathcal{A}^{l}}^{u,k} = 0$ for some $k$.
We know that if $\Phi_{\mathcal{L}^s}^u = 0$, then
\begin{gather}
    \mathcal{L}^{s}\big([\delta,\mu,\tau];\theta_{\mathcal{L}^s}\big) \;\neq\; [\,\hat{\mathcal{L}}(\delta,\mu;\theta_{\hat{\mathcal{L}}}), \tau\,], \quad \tau \in \mathbb{R}^{s_t} 
\end{gather}
which is a contradiction of Property \ref{prop:leaf_uniquenes} if $u \in U$. Therefore $\Phi_{\mathcal{L}^s}^u = 1$. 

Similarly, we know that if $\Phi_{\mathcal{A}^{l}}^{u,k} = 0$ for some $k$, $u\in U$ then this is a violation of Property \ref{prop:mer_con} which implies an inconsistent Merkle path. Further using Property \ref{prop:leaf_path} we can also assert that the given user $u\in U$ corresponds to a consistent Merkle path which presents a contradiction. We can ultimately establish the fact that if $u \in U$ then it necessarily implies that $\Phi_{\mathcal{L}^s}^u = 1 \text{ and }\ \Phi_{\mathcal{A}^{l}}^{u,k} = 1$.

Case 2.2 Assume we have a user $u\notin U$, with $\mathcal{M}([u,\delta]|\mathcal{R}(\mathcal{A},\mathcal{L},U))=0$. As before, we suppose that either $\Phi_{\mathcal{L}^s}^u = 0 \text{ or }\ \Phi_{\mathcal{A}^{l}}^{u,k} = 0$ for some $k$. Since $u\notin U$, this implies that $ \mathcal{L}^{s}\big([\delta,\mu,\tau];\theta_{\mathcal{L}^s}\big) = \mathcal{L}^s(\varnothing)$. However, we also know that if $\Phi_{\mathcal{L}^s}^u = 0$, then
\begin{gather}
    \mathcal{L}^{s}\big([\delta,\mu,\tau];\theta_{\mathcal{L}^s}\big) \;\neq\; [\,\hat{\mathcal{L}}(\delta,\mu;\theta_{\hat{\mathcal{L}}}), \tau\,], \quad \tau \in \mathbb{R}^{s_t} 
\end{gather}

This implies that we encounter a contradiction
\begin{gather}
    \mathcal{L}^{s}\big([\delta,\mu,\tau];\theta_{\mathcal{L}^s}\big) \;\neq\ \mathcal{L}^s(\varnothing) \implies u\in U
\end{gather}

Similarly, we know that if $\Phi_{\mathcal{A}^{l}}^{u,k} = 0$ for some $k$ $u\notin U$ then this is a violation of Property \ref{prop:mer_con} which implies an inconsistent Merkle path. As before, using Property \ref{prop:leaf_path} we can also assert that even when user $u\notin U$ corresponds to a consistent Merkle path. This implies a contradiction establishing the fact that if $u \notin U$ it necessarily implies that $\Phi_{\mathcal{L}^s}^u = 1 \text{ and }\ \Phi_{\mathcal{A}^{l}}^{u,k} = 1$.
\end{proof}

\subsection{Proof of Proposition \ref{prs:ks_lrt_mrp}}\label{subsec:proof_ks_lrt_mrp}
\begin{proof}
We already know under Proposition \ref{prs:pi_mem} that verifying the LTR and MRP proofs is a necessary and sufficient condition for realizing the membership of user $u$ in a reduction operation given $R(\mathcal{A},\mathcal{L},U^{inc})$. 

Let $\Omega^u_\delta$ denote the raw user data record $\delta^u$ of user $u$ as contained in witness $\Omega^u_{LTR} \in \Omega^u$. 
We assume a compromised CRO possessing an extractor $\mathcal{E}^{\mathcal{C}}$ that produces a valid proof $\hat{\Pi}^u$ with an incorrect witness $\hat{\Omega}^u = \mathcal{E}^{\mathcal{C}}(\Omega^{u})$ such that
\begin{gather}\label{eq:verinc}
    \hat{\Phi}^{(\delta,\mu)}= \mathsf{VerInc}(H^{(\delta,\mu)},H^{leaf},\hat{\Pi}^{u},H^{root},H^\eta,vk_{\mathcal{L}^s},vk_{\mathcal{A}^l}) = 1
\end{gather}

There are two possible outcomes to generate an incorrect witness $\hat{\Omega}^u$


\noindent Case 1 (Incorrect salted raw data): 
In this case, we consider a scenario wherein a compromised CRO tries to generate a valid zk-SNARK $\hat\Pi^u$ that maintains the consistency between the first input in $\hat\Pi^u$ and the salted raw user data hash $H^{(\delta,\mu)}$. 

More technically, the CRO tries to generate a valid $\hat\Pi^u_{\mathcal{L}^s}$ using an incorrect witness $\hat\Omega^u=\{\hat\Omega^u_{LTR},\Omega^u_{MRP}\}$.
To prove using contradiction, we assume there exists a perturbed salted raw data tuple $\hat\Omega^u_{LTR} = \{\hat\delta_u,\hat\mu_u,\tau_u\}$ that leads to $\hat{\Phi}^{(\delta,\mu)}=1$. We denote the LTR constraint set $\mathcal{Z}^u_{LTR}$ for user u as defined by Equation system \eqref{eq:verinc4}.
\begin{gather}
    \hat{\Pi}_{\mathcal{L}^s}[\mathsf{Input1}] = H^{(\delta,\mu)},\ \hat{\Pi}_{\mathcal{L}^s}[\mathsf{Input2}] = H^{\tau}, \ \hat{\Pi}_{\mathcal{L}^s}[\mathsf{Output}] = H^{leaf}_u \label{eq:verinc4}
\end{gather}
We know from Algorithm \ref{alg:csmt_ltr_proof_verify} that $(H^{(\delta,\mu)},H^\tau,H^{leaf}_u) \in \mathcal{Z}^{u}_{LTR}$
must hold for \texttt{LTRVerify} to return a success as given in Equation \eqref{eq:verinc0}.  
\begin{gather}
\mathsf{LTRVerify}(vk_{\mathcal{L}^s}, \Pi_{\mathcal{L}^s}^u, H^{(\delta_u,\mu_u)}, H^{leaf}_u) = 1 \label{eq:verinc0}
\end{gather}


By the knowledge-soundness of Halo2~\cite{halo2_zcash_2023}, a valid $\hat\Pi_{\mathcal{L}^s}$ implies the existence of an extracted witness satisfying the circuit relation $\mathcal{R}_{\mathcal{L}^s}$, except with probability $\mathrm{negl}_1(\lambda)$. Since $H$ is collision-resistant, a witness satisfying the in-circuit hash constraints $H(\hat\delta_u,\hat\mu_u)=H^{(\delta,\mu)}$ and $H(\tau_u)=H^\tau$ implies $(\hat\delta_u,\hat\mu_u,\tau_u)=(\delta_u,\mu_u,\tau_u)$, except with probability $\mathrm{negl}_2(\lambda)$. Hence
\[
\Pr\big[\mathsf{Verify}(vk_{\mathcal{L}^s},\hat\Pi_{\mathcal{L}^s})=1 \wedge \hat\Omega^u_{LTR}\ne\Omega^u_{LTR}\big] \le \mathrm{negl}_1(\lambda)+\mathrm{negl}_2(\lambda) = \mathrm{negl}(\lambda)
\]

This leads us to a contradiction thereby establishing the fact that a compromised CRO cannot generate a valid zk-SNARK that proves the inclusion of the incorrect data for any user. 

\noindent Case 2 (Incorrect MRP proof): We consider a scenario where a compromised CRO generates a valid proof $\hat\Pi^u_{\mathcal{A}^l}$ using an incorrect MRP witness $\hat\Omega^u_{MRP}$. Using contradiction, we assume that there exists a valid zk-SNARK $\hat\Pi^{u,k}_{\mathcal{A}^l}\in$ that is generated using an incorrect $\hat\Omega^{u,k}_{MRP}$. 

\noindent Case 2.1 (Incorrect Selector Bit Path): In this case, to prove using contradiction, we assume that the CRO supplies $\hat\Pi^u_{\mathcal{A}^l}$ that corresponds to an
incorrect salted leaf transform $\hat{H}^{leaf}_u$ instead of the genuine $H^{leaf}_u$ for user $u$. However, we know that $\hat H^{leaf}_u$ cannot lead to a successful evaluation of \texttt{LTRVerify} at the verifier level due to the knowledge-soundness argument of Case 1. As a consequence, by contradiction we can say that the CRO cannot use an incorrect selector bit path sequence $\hat B$ instead of the correct selector bit path $B$ without violating the knowledge-soundness property.

\noindent Case 2.2 (Incorrect MRP hop inputs): 
We assume that the CRO uses an incorrect hop witness for some level $1 \leq k \leq K$ while using the correct selector bit path index $B$. For a particular level $k$, we denote MRP constraint set as $\mathcal{Z}^{k,u}_{MRP}$ which is defined by Equations \eqref{eq:verinc10}-\eqref{eq:verinc11}.
\begin{gather}
    \mathsf{hash}(B[k]) = H^{k,b} \text{ and } \mathsf{Bin}(H^{leaf}_u) = B\label{eq:verinc10}\\
    H^{k,L} = 
    \begin{cases}
       \hat\Pi^{k-1}_{\mathcal{A}^l}[\mathsf{Parent}], \ \text{ if } B[k] = 0\\
       \hat\Pi^k_{\mathcal{A}^l}[\mathsf{LeftInput}], \ \text{ otherwise }
    \end{cases}    
    \label{eq:verinc7}\\ 
    H^{k,R} = 
    \begin{cases}
       \hat\Pi_{\mathcal{A}^l}[\mathsf{RightInput}], \ \text{ if } B[k] = 0\\
       \hat\Pi^{k-1}_{\mathcal{A}^l}[\mathsf{Parent}] \ \text{ otherwise }
    \end{cases}    
    \label{eq:verinc8}\\ 
    \hat\Pi_{\mathcal{A}^l}[\mathsf{Bit}] = H^{k,b}\label{eq:verinc9}\\
    \hat\Pi_{\mathcal{A}^l}[\mathsf{Nonce}] = H^{\eta}\label{eq:verinc11}
\end{gather}
We know from Algorithm \ref{alg:csmt_mrp_proof_verify} that for a particular level $k$ Equation \eqref{eq:verinc5} must be fulfilled to successfully validate MRP hop proof artifact.
\begin{gather}
    \mathsf{MRPHopVerify}(vk_{\mathcal{A}^{l}},\hat\Pi_{\mathcal{A}^{l}}^k,H^{k,L},H^{k,R},H^{k,b},H^{\eta}) = 1 \label{eq:verinc5}
\end{gather}
\noindent By applying the knowledge-soundness and collision-resistance argument we know that the following condition holds.
\[
\begin{aligned}
\Pr&\big[
\verb|Verify|(vk_{\mathcal{A}^l}, \hat\Pi^u_{\mathcal{A}^l}) = 1 \wedge\ (H^{k,L},H^{k,R},H^{k,b},H^{\eta}) \in \mathcal{Z}^{k,u}_{MRP} 
\big]
\ \le\ \text{negl}(\lambda)
\end{aligned}
\]
This is a direct contradiction of our assumption which completes the proof.
\end{proof}



\subsection{Proof of Proposition \ref{prs:exc_check}}\label{subsec:proof_exc_check}
\begin{proof}
Consider a subset $\mathcal{D}_q \subseteq\Delta^{\mathcal{L}^s}_U$ of $q$ ordered leaf nodes, where $\mathcal{D}_q = \{\varphi^{0,s}_{1},\varphi^{0,s}_{2} \ldots \varphi^{0,s}_{q}\}$ and $q\leq 2^K$. Further consider the lower and upper bound decimal representation limits of the hash of the individual leaves of set $\mathcal{D}_q$ defined in Equations \eqref{eq:exc_check0} and \eqref{eq:exc_check0b}.
\begin{gather}
    L^q = min\Big( \{\texttt{Decimal}[\texttt{hash}(\varphi^{0,s}_{1})], \ldots \texttt{Decimal}[\texttt{hash}(\varphi^{0,s}_{q})] \} \Big) \label{eq:exc_check0}\\
    U^q = max\Big( \{\texttt{Decimal}[\texttt{hash}(\varphi^{0,s}_{1})], \ldots \texttt{Decimal}[\texttt{hash}(\varphi^{0,s}_{q})] \} \Big) \label{eq:exc_check0b}
\end{gather}

Let $H^{root}_{\mathcal{D}_s}$ be the subtree root corresponding to the ordered transformed leaf set $\mathcal{D}_q$. Now consider a $\hat{\phi}$, such that we obtain another ordered set $\hat{\mathcal{D}}_{q} = \mathcal{D}_s \underset{ord}{\cup} \{\hat{\phi}\}$ with $H^{root}_{\hat{\mathcal{D}}_s}$ as the subtree root of $\hat{\mathcal{D}}_q$. 

To prove by contradiction, let us assume that $H^{root}_{\hat{\mathcal{D}}_q} = H^{root}_{\mathcal{D}_q}$. We know that we can define $\hat{H} = \texttt{hash}(\hat{\phi})$ for which 
\begin{gather}
L^q\leq \texttt{Decimal}[\hat{H}] \leq U^q, \text{ and } H^{root}_{\hat{\mathcal{D}}_q} = H^{root}_{\mathcal{D}_q}
\end{gather} 

However, using Property \ref{prop:leaf_uniquenes}, we know that all leaf transformations in $\hat{\mathcal{D}}_{q}$ must be unique. Therefore, using the computational aggregation primitive $\mathcal{A}^s(.;\theta_{\mathcal{A}^s})$ as defined in Equation \eqref{eq:a_rec} we obtain Equation \eqref{eq:exc_check_1}.
\begin{equation}\label{eq:exc_check_1}
    \mathcal{A}^s(\mathcal{D}_q;\theta_{\mathcal{A}^s}) \neq \mathcal{A}^s(\mathcal{\hat{\mathcal{D}}_q;\theta_{\mathcal{A}^s}})
\end{equation}
As a consequence of Equation \eqref{eq:exc_check_1} along with the properties of the hash function defined in Definition \ref{defn3}, we can say that $H^{root}_{\mathcal{D}_q} \neq H^{root}_{\hat{\mathcal{D}}_q}$. As a result, we violate Definition \ref{def:csmt} which governs the construction of the CSMT, thereby completing the proof.


\end{proof}

\section{Supplementary Text}
\subsection{Case Studies}
\subsubsection{Huntington Disease} \label{apx:hd}
Huntington's disease (HD) results from CAG repeat expansion in the HTT gene \cite{Macdonald1993_mp}, with repeat number correlating with onset and severity \cite{Lee2012_hn}. We tested whether CAG repeat distributions differ between 12 healthy individuals and 12 HD-symptomatic individuals \cite{Moily2014_wn, Jiang2014_we, Gardiner2017_vg, Vater2025_fq} using a two-sample Kolmogorov-Smirnov test, with our framework providing membership and non-membership proofs without disclosing participant data.

\subsubsection{HIV-1} \label{apx:hiv}
HIV-1 resistance to the Env-targeting therapeutic Temsavir (TMR) \cite{Markham2020, Wang2018_xk} is driven largely by mutations at Env positions 375, 426, 434, and 475 \cite{Gartland2021_mz, Zuze2023_wh}. We used training (n=479) and test (n=85) sets with 19 dummy-encoded sequence features per sample and PhenoSense GT-measured resistance, demonstrating sequence-based resistance prediction under our privacy-preserving verification framework.

\section{Supplementary Algorithms}    \label{apx:algorithms}
\subsection{Algorithm for detecting data exclusivity}\label{subsec:data_exc_exp}
Algorithm~\ref{alg:csmt_data_exclusivity} defines \verb|VerifyDataExclusivity|, which takes the PHR Merkle root reference $\mathcal{P}$, the set of non-default salted leaf nodes $\mathcal{T}$, and, for each user $u$ in the study, the hash tuple $(H^{(\delta_u,\mu_u)}, H^{\tau_u})$, collectively denoted $\mathcal{H}$. By Assumption~\ref{as:phr}, every hash tuple in $\mathcal{H}$ must be substantiated by a Merkle proof from $\mathcal{P}$; both sets are pre-committed by the CRO ahead of Algorithm~\ref{alg:buildsmt}. The algorithm therefore first checks PHR Merkle consistency, $(H^{(\delta_u,\mu_u)},H^{\tau_u}) \in \mathcal{P} \cap \mathcal{H}$, for each LTR proof, building a set of non-default leaf hashes as the basis for the MRP exclusivity check. It then performs a depth-first traversal of the CSMT to collect the nodes reachable from non-default leaves. Per Definition~\ref{def:data_exc} and Proposition~\ref{prs:exc_check}, data exclusivity holds when (i) the leaf set found by MRP traversal matches the leaf set found from the LTR proofs, and (ii) every non-leaf MRP node has at least one leaf descendant.

\begin{algorithm}[htbp]
 \caption{Data Exclusivity Check}
 \label{alg:csmt_data_exclusivity}
 \begin{algorithmic}[1]
    \Function{VerifyDataExclusivity}{$\mathcal{H}^{inc},\mathcal{P},\mathcal{T},\Pi^{U^{inc}}_{\mathsf{CSMT}}$}
        \State $\mathsf{NodesLTR} \leftarrow \varnothing$
        \For{$\Pi^u_{\mathcal{L}^{s}} \in \Pi^{U^{inc}}_{\mathsf{CSMT}}$}
            \State retrieve $H^{leaf}_u, H^{(\delta_u,\mu_u)}, H^{\tau_u}$ from $\Pi^u_{\mathcal{L}^{s}}$
            \If{$(H^{(\delta_u,\mu_u)},H^{\tau_u}) \in \mathcal{P} \cap \mathcal{H}$}
                \State $\mathsf{NodesLTR} \leftarrow \mathsf{NodesLTR} \cup \{H^{leaf}_u\}$
            \EndIf
        \EndFor
        \State $\mathsf{NodesMRP} \leftarrow \varnothing$
        \For{$H^{leaf} \in \{\mathcal{T}:\mathsf{SaltedLeafSet}\}$, $k = 1$ \textbf{to} $K$}
            \State retrieve $v_{sib}, v$ from $\Pi_{\mathcal{A}^{l}}^k \leftarrow \Pi^{u}_{\mathcal{A}^l}[k]$
            \State $\mathsf{NodesMRP}[v_{sib}].\mathsf{Leaves} \leftarrow \mathsf{NodesMRP}[v_{sib}].\mathsf{Leaves} \cup \varnothing$
            \State $\mathsf{NodesMRP}[v].\mathsf{Leaves} \leftarrow \mathsf{NodesMRP}[v].\mathsf{Leaves} \cup \{H^{leaf}\}$
        \EndFor
        \If{$\bigcup_{v}\mathsf{NodesMRP}[v].\mathsf{Leaves} \neq \mathsf{NodesLTR}$}
        \State \Return spurious leaf existence detected
        \EndIf
        \For{$v \in \mathsf{NodesMRP}$}
        \If{$\mathsf{NodesMRP}[v].\mathsf{Leaves} = \varnothing$ and $v\notin \mathcal{T}$}
        \State \Return spurious leaf existence detected
        \EndIf
        \EndFor
    \EndFunction
 \end{algorithmic}
\end{algorithm}

\subsection{Verifying zk-SNARK Artifacts} \label{apx:snarkver}
Algorithm \ref{alg:csmt_ltr_proof_verify} defines function \texttt{LTRVerify} as a means to verify LTR proofs based on verification key $vk_{\mathcal{L}^s}$, LTR zk-SNARK $\Pi_{\mathcal{L}^s}$, hash of raw salted user data, transform salt and salted leaf transform $H^{raw}$,$H^{\tau}$,$H^{leaf}$ respectively. The \texttt{LTRVerify} function primarily validates the given zk-SNARK using the verification key using the \texttt{Verify} function. Additionally, it ensures hash consistency between supplied hashes $H^{raw}$,$H^{\tau}$,$H^{leaf}$ and the input and output hashes contained in the provided zk-SNARK. The final verification output of the function depends on successfully passing the criteria for zk-SNARK validation as well as the hash consistency check.

\begin{algorithm}[htbp]
 \caption{Verifying LTR Proof}
 \label{alg:csmt_ltr_proof_verify}
 \begin{algorithmic}[1]
    \Function{LTRVerify}{$vk_{\mathcal{L}^s}, \Pi_{\mathcal{L}^s}, H^{raw},H^{\tau}, H^{leaf}$}
        \State $\Phi_{zk}\leftarrow \mathsf{Verify}(vk_{\mathcal{L}^s}, \Pi_{\mathcal{L}^s})$
        \For{$(\mathsf{Field}, H) \in \{(\mathsf{Input1},H^{raw}), (\mathsf{Input2},H^{\tau}), (\mathsf{Output},H^{leaf})\}$}
        \If{$\Pi_{\mathcal{L}^s}[\mathsf{Field}] = H$}
        \State $\Phi_{\mathsf{Field}} \leftarrow 1$
        \Else
        \State $\Phi_{\mathsf{Field}} \leftarrow 0$
        \EndIf
        \EndFor
        \State \Return $\Phi_{zk} \cap \Phi_{us} \cap \Phi_{ts} \cap \Phi_{output}$
    \EndFunction
\end{algorithmic}
\end{algorithm}
In Algorithm \ref{alg:csmt_mrp_proof_verify} we present function \texttt{MRPHopVerify} which handles the verification of per-hop MRP proofs. The \texttt{MRPHopVerify} function consumes the verification key for aggregator circuit, zk-SNARK proof artifact as well as the hashes of the left and right sibling, selector bit path index array and the nonce. These input variables are denoted by $vk_{\mathcal{A}^{l}},\Pi_{\mathcal{A}^{l}},H^{L},H^{R},H^b,H^\eta$ respectively. The function validates the provided zk-SNARK using the verification key and checks for hash consistency of left and right siblings, selector bits and the nonce. Passing each criteria results in an overall success of the \texttt{MRPHopVerify} function.

\begin{algorithm}[htbp]
 \caption{Verifying MRP Hop Proofs}
 \label{alg:csmt_mrp_proof_verify}
 \begin{algorithmic}[1]
    \Function{MRPHopVerify}{$vk_{\mathcal{A}^{l}},\Pi_{\mathcal{A}^{l}},H^{L},H^{R},H^b,H^\eta$}
        \State $\Phi_{zk}\leftarrow \mathsf{Verify}(vk_{\mathcal{A}^s}, \Pi_{\mathcal{A}^s})$
        \For{$(\mathsf{Field}, H) \in \{(\mathsf{LeftInput},H^L), (\mathsf{RightInput},H^R), (\mathsf{Bit},H^b), (\mathsf{Nonce},H^\eta)\}$}
        \If{$\Pi_{\mathcal{A}^s}[\mathsf{Field}] = H$}
        \State $\Phi_{\mathsf{Field}} \leftarrow 1$
        \Else
        \State $\Phi_{\mathsf{Field}} \leftarrow 0$
        \EndIf
        \EndFor
        \State \Return $\Phi_{zk} \cap \Phi_{L} \cap \Phi_{R} \cap \Phi_{b} \cap \Phi_{\eta}$
    \EndFunction
\end{algorithmic}
\end{algorithm}


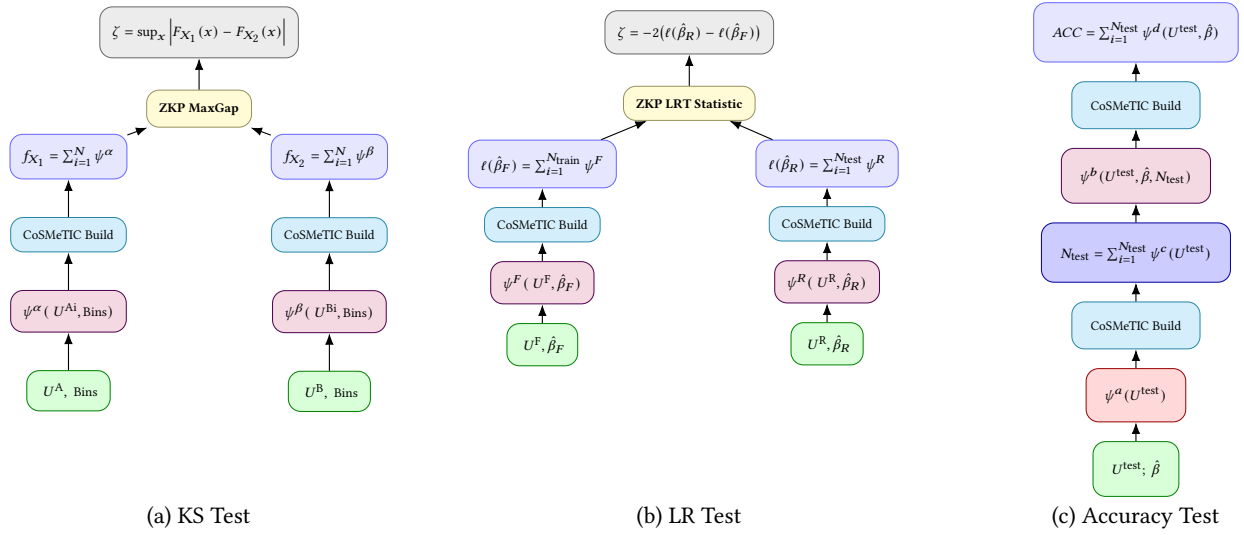
\begin{figure*}[t]
\centering
\setlength{\tabcolsep}{6pt}

\begin{tabular}{ccc}
\begin{minipage}[t]{0.31\textwidth}
\vspace{0pt}
\centering
\resizebox{\linewidth}{!}{%
  \begin{tikzpicture}[font=\scriptsize, >=Latex,
    baseline=(current bounding box.north),
    arrow/.style={-Latex, line width=0.5pt},
    box/.style={draw, rounded corners=2mm, align=center, inner sep=2mm},
    top/.style={box, fill=gray!15, draw=black!60},
    agg/.style={box, fill=yellow!20, draw=yellow!50!black},
    proc/.style={box, fill=blue!10, draw=blue!60},
    smt/.style={box, fill=cyan!15, draw=cyan!60!black},
    leaf/.style={box, fill=purple!15, draw=violet!60!black},
    input/.style={box, fill=green!15, draw=green!50!black}
  ]
  \path[use as bounding box] (-3.3,0.4) rectangle (3.3,-7.5);
  \node[top] (zeta)
  {$\zeta = \sup_x \left|F_{X_1}(x)-F_{X_2}(x)\right|$};

  \node[agg, below=5mm of zeta] (maxgap)
  {\textbf{ZKP MaxGap}};

  \node[proc, below left=1mm and 3mm of maxgap] (fx1)
  {$f_{X_1}=\sum_{i=1}^N \psi^{\alpha}$};

  \node[proc, below right=1mm and 3mm of maxgap] (fx2)
  {$f_{X_2}=\sum_{i=1}^N \psi^{\beta}$};

  \node[smt, below=6mm of fx1] (smt1)
  {CoSMeTIC Build};

  \node[smt, below=6mm of fx2] (smt2)
  {CoSMeTIC Build};

  \node[leaf, below=6mm of smt1] (psi1)
  {$\psi^{\alpha}(\;U^{\mathrm{Ai}},\text{Bins})$};

  \node[leaf, below=6mm of smt2] (psi2)
  {$\psi^{\beta}(\;U^{\mathrm{Bi}},\text{Bins})$};

  \node[input, below=6mm of psi1] (x1)
  {$\;U^{\mathrm{A}},\;\text{Bins}$};

  \node[input, below=6mm of psi2] (x2)
  {$\;U^{\mathrm{B}},\;\text{Bins}$};

  \draw[arrow] (maxgap) -- (zeta);
  \draw[arrow] (fx1) -- (maxgap);
  \draw[arrow] (fx2) -- (maxgap);
  \draw[arrow] (smt1) -- (fx1);
  \draw[arrow] (smt2) -- (fx2);
  \draw[arrow] (psi1) -- (smt1);
  \draw[arrow] (psi2) -- (smt2);
  \draw[arrow] (x1) -- (psi1);
  \draw[arrow] (x2) -- (psi2);

  \end{tikzpicture}%
}
\vspace{1mm}
{(a) KS Test}
\end{minipage}
\hspace{5mm}
&
\begin{minipage}[t]{0.31\textwidth}
\vspace{0pt}
\centering
\resizebox{\linewidth}{!}{%
  \begin{tikzpicture}[font=\scriptsize, >=Latex,
    baseline=(current bounding box.north),
    arrow/.style={-Latex, line width=0.5pt},
    box/.style={draw, rounded corners=2mm, align=center, inner sep=2mm},
    top/.style={box, fill=gray!15, draw=black!60},
    agg/.style={box, fill=yellow!20, draw=yellow!50!black},
    proc/.style={box, fill=blue!10, draw=blue!60},
    smt/.style={box, fill=cyan!15, draw=cyan!60!black},
    leaf/.style={box, fill=purple!15, draw=violet!60!black},
    input/.style={box, fill=green!15, draw=green!50!black}
  ]
  \path[use as bounding box] (-3.3,0.4) rectangle (3.3,-7.5);

  \node[top] (lrt)
  {$\zeta  = -2\big(\ell(\hat{\beta}_R) - \ell(\hat{\beta}_F)\big)$};

  \node[agg, below=5mm of lrt] (lrtcirc)
  {\textbf{ZKP LRT Statistic}};

  \node[proc, below left=3mm and 1mm of lrtcirc] (lfull)
  {$\ell(\hat{\beta}_F)=\sum_{i=1}^{N_{\text{train}}}\psi^{F}$};

  \node[proc, below right=3mm and 0mm of lrtcirc] (lred)
  {$\ell(\hat{\beta}_R)=\sum_{i=1}^{N_{\text{test}}}\psi^{R}$};

  \node[smt, below=3mm of lfull] (smtF) {CoSMeTIC Build};
  \node[smt, below=3mm of lred]  (smtR) {CoSMeTIC Build};

  \node[leaf, below=3mm of smtF] (psiF)
  {$\psi^{F}(\;U^{\mathrm{F}},\hat{\beta}_F)$};

  \node[leaf, below=3mm of smtR] (psiR)
  {$\psi^{R}(\;U^{\mathrm{R}},\hat{\beta}_R)$};

  \node[input, below=3mm of psiF] (inputF)
  {$\;U^{\mathrm{F}},\hat{\beta}_F$};

  \node[input, below=3mm of psiR] (inputR)
  {$\;U^{\mathrm{R}},\hat{\beta}_R$};

  \draw[arrow] (lrtcirc) -- (lrt);
  \draw[arrow] (lfull) -- (lrtcirc);
  \draw[arrow] (lred)  -- (lrtcirc);
  \draw[arrow] (smtF)  -- (lfull);
  \draw[arrow] (smtR)  -- (lred);
  \draw[arrow] (psiF)  -- (smtF);
  \draw[arrow] (psiR)  -- (smtR);
  \draw[arrow] (inputF) -- (psiF);
  \draw[arrow] (inputR) -- (psiR);

  \end{tikzpicture}%
}
\vspace{1mm}
{(b) LR Test}
\end{minipage}
&
\begin{minipage}[t]{0.31\textwidth}
\vspace{0pt}
\centering
\resizebox{\linewidth}{!}{%
  \begin{tikzpicture}[font=\scriptsize, >=Latex,
    baseline=(current bounding box.north),
    arrow/.style={-Latex, line width=0.5pt},
    box/.style={draw, rounded corners=2mm, align=center, inner sep=3mm},
    top/.style={box, fill=blue!10, draw=blue!60},
    smt/.style={box, fill=cyan!15, draw=cyan!60!black},
    midpurple/.style={box, fill=purple!15, draw=violet!60!black},
    midblue/.style={box, fill=blue!20, draw=blue!60!black},
    midred/.style={box, fill=red!15, draw=red!60!black},
    input/.style={box, fill=green!15, draw=green!50!black}
  ]
  \path[use as bounding box] (-3.3,0.4) rectangle (3.3,-7.5);
  \node[top] (accF)
  {$ACC=\sum_{i=1}^{N_{\mathrm{test}}}\psi^{d}(U^{\mathrm{test}},\hat{\beta})$};

  \node[smt, below=3mm of accF] (smtTop) {CoSMeTIC Build};

  \node[midpurple, below=3mm of smtTop] (psib)
  {$\psi^{b}(U^{\mathrm{test}},\hat{\beta},N_{\mathrm{test}})$};

  \node[midblue, below=3mm of psib] (Ntest)
  {$N_{\mathrm{test}}=\sum_{i=1}^{N_{\mathrm{test}}}\psi^{c}(U^{\mathrm{test}})$};

  \node[smt, below=3mm of Ntest] (smtBottom) {CoSMeTIC Build};

  \node[midred, below=3mm of smtBottom] (psia)
  {$\psi^{a}(U^{\mathrm{test}})$};

  \node[input, below=3mm of psia] (inputs)
  {$U^{\mathrm{test}};\ \hat{\beta}$};

  \draw[arrow] (smtTop) -- (accF);
  \draw[arrow] (psib) -- (smtTop);
  \draw[arrow] (Ntest) -- (psib);
  \draw[arrow] (smtBottom) -- (Ntest);
  \draw[arrow] (psia) -- (smtBottom);
  \draw[arrow] (inputs) -- (psia);

  \end{tikzpicture}%
}
\vspace{1mm}
{(c) Accuracy Test}
\end{minipage}
\end{tabular}

\caption{CoSMeTIC workflows for KS, LRT, and accuracy.}
\label{fig:three_workflows_row}
\end{figure*}

\subsection{Kolmogorov--Smirnov (KS) Test}\label{subsec:ks-test}
We first consider the two-sample Kolmogorov--Smirnov (KS) test, which compares the empirical distributions of Group~A (healthy) and Group~B (HD) (Algorithm~\ref{alg:ks_algorithm}). The leaf-level transformation $\mathcal{L}^{s}_{BC}$ maps each data point to a one-hot vector indicating its bin assignment under a pre-specified bin vector $\theta_{\mathcal{L}^{s}_{BC}}$, and the aggregation function $\mathcal{A}^{l}_{sum}$ sums these vectors across observations within each group, yielding bin-count vectors $\Psi^A$, $\Psi^B$. The \textsc{MaxAbsoluteGap} circuit then converts $\Psi^A$, $\Psi^B$ into cumulative distributions $F_A$, $F_B$ and returns $\zeta = \|F_A - F_B\|_{\infty}$, the standard KS statistic measuring the largest gap between the two empirical distributions. The KS verifier confirms correctness by checking the CSMT-based inclusion/exclusion proof for $u$, the max-gap proof $\Pi_{MAG}$ against $vk_{MAG}$, and, if both pass, that the sample-group roots $H^{root,A}$, $H^{root,B}$ match their expected hashes; it returns failure tagged with whichever check did not hold.


\begin{algorithm}[htbp]
 \caption{Two-sample KS Test Statistic Computation}
 \label{alg:ks_algorithm}
 \begin{algorithmic}[1]
    \Function{KS2Sample}{$U,U^{A},U^B$}
        \Statex\LineComment{setup bin count leaf transformation circuit}
        \State $pk_{\mathcal{L}^{s}_{BC}},vk_{\mathcal{L}^{s}_{BC}}\leftarrow\mathsf{Setup}(1^\lambda,\mathcal{L}^{s}_{BC},\theta_{\mathcal{L}^{s}_{BC}})$

        \Statex\LineComment{setup bin count vector sum aggregation circuit}
        \State $pk_{\mathcal{A}^{l}_{sum}},vk_{\mathcal{A}^{l}_{sum}}\leftarrow\mathsf{Setup}(1^\lambda,\mathcal{A}^{l}_{sum},\theta_{\mathcal{A}^{l}_{sum}})$

        \Statex \LineComment{salted bin count \cosmetic build on Group A}
        \State $\Psi^A,H^{root,A}\leftarrow$\Call{CosmeticCROBuild}{$U,U^A,vk_{\mathcal{L}^{s}_{BC}},vk_{\mathcal{A}^{l}_{sum}}$}

        \Statex \LineComment{salted bin count \cosmetic build on Group B}
        \State $\Psi^B,H^{root,B}\leftarrow$\Call{CosmeticCROBuild}{$U,U^B,vk_{\mathcal{L}^{s}_{BC}},vk_{\mathcal{A}^{l}_{sum}}$}

        \Statex\LineComment{setup post-aggregation Max Absolute Gap circuit}
        \State $pk_{MAG},vk_{MAG}\leftarrow\mathsf{Setup}(1^\lambda,\mathcal{G}_{MAG})$
        \State distribute $vk_{MAG}$ to all users in $U^A$ and $U^B$
        \State $\zeta\leftarrow$\Call{MaxAbsoluteGap}{$\Psi^A,\Psi^B$}
        \State set $\Omega_{MAG} = \{\Psi^A,\Psi^B,\zeta\}$
        \State generate $\Pi_{MAG} = \mathsf{Prove}(pk_{MAG},\Omega_{MAG})$
        \State distribute $\Pi_{MAG},pk_{MAG}$ publicly
    \EndFunction
 \end{algorithmic}
\end{algorithm}




\subsection{Logistic Likelihood-Ratio Test}\label{subsec:lrt_test}
We next evaluate the Logistic Likelihood-Ratio Test (LRT), which assesses the statistical significance of a full logistic regression model relative to a reduced baseline (Algorithm~\ref{alg:lrt_algorithm}). The leaf-level transformation is the log-likelihood function $\mathcal{L}^{s}_{LL}$, parameterized by the regression coefficients $\theta_{\mathcal{L}^{s}_{LL}}$ and aggregated via the sum aggregator $\mathcal{A}^{l}_{sum}$; full and reduced models are evaluated independently within \cosmetic, yielding summed log-likelihoods $\Psi^{r}$, $\Psi^{f}$ that a zero-knowledge circuit combines into the LRT statistic $\zeta = -2(\Psi^{r} - \Psi^{f})$. The LRT verifier confirms correctness by checking the CSMT-based inclusion/exclusion proof for $u$, the log-likelihood proof $\Pi_{LL}$ against $vk_{LL}$, and, if both pass, that the full and reduced model roots $H^{root,f}$, $H^{root,r}$ match their expected hashes; it returns failure tagged with whichever check did not hold.


\begin{algorithm}[htbp]
 \caption{Likelihood Ratio Test Computation}
 \label{alg:lrt_algorithm}
 \begin{algorithmic}[1]
    \Function{LikelihoodRatioTest}{$U,U^{f},U^r$}
    
        \Statex \LineComment{setup log likelihood leaf transformation circuit}
        \State $pk_{\mathcal{L}^{s}_{LL}}, vk_{\mathcal{L}^{s}_{LL}} \leftarrow \mathsf{Setup}
        (1^\lambda, \mathcal{L}^{s}_{LL}, \theta_{\mathcal{L}^{s}_{LL}})$

        \Statex \LineComment{setup log likelihood scalar sum aggregation circuit}
        \State $pk_{\mathcal{A}^{l}_{sum}}, vk_{\mathcal{A}^{l}_{sum}} \leftarrow \mathsf{Setup}(1^\lambda, \mathcal{A}^{l}_{sum}, \theta_{\mathcal{A}^{l}_{sum}})$

        \Statex \LineComment{salted log likelihood \cosmetic build on full model}
        \State $\Psi^f, H^{root, f} \leftarrow$ \Call{CosmeticCROBuild}{$U, U^f, vk_{\mathcal{L}^{s}_{LL}}, vk_{\mathcal{A}^{l}_{sum}}$}

        \Statex \LineComment{salted log likelihood \cosmetic build on reduced model}
        \State $\Psi^r, H^{root, r} \leftarrow$ \Call{CosmeticCROBuild}{$U, U^r, vk_{\mathcal{L}^{s}_{LL}}, vk_{\mathcal{A}^{l}_{sum}}$}

        \Statex \LineComment{setup post-aggregation LRT statistic circuit}
        \State $pk_{LRT},vk_{LRT} \leftarrow \mathsf{Setup}(1^\lambda, \mathcal{G}_{LRT})$
        \State distribute $vk_{LRT}$ to all users in $U^f$ or $U^r$
        \State $\zeta \leftarrow$ \Call{LRTStatistic}{$\Psi^f, \Psi^r$}
        \State set $\Omega_{LRT} = \{\Psi^f, \Psi^r, \zeta\}$
        \State generate $\Pi_{LRT} = \mathsf{Prove}(pk_{LRT},\Omega_{LRT})$
        \State distribute $\Pi_{LRT}, pk_{LRT}$ publicly
    \EndFunction
 \end{algorithmic}
\end{algorithm}

\begin{algorithm}[htbp]
 \caption{Accuracy Computation}
 \label{alg:acc_algorithm}
 \begin{algorithmic}[1]
    \Function{Accuracy}{$U,U^{test}$}
    
        \LineComment{setup classification assessment leaf transformation circuit}
        \State $pk_{\mathcal{L}^{s}_{CA}}, vk_{\mathcal{L}^{s}_{CA}} \leftarrow \mathsf{Setup}
        (1^\lambda, \mathcal{L}^{s}_{CA}, \theta_{\mathcal{L}^{s}_{CA}})$

        \LineComment{setup classification assessment sum aggregation circuit}
        \State $pk_{\mathcal{A}^{l}_{sum}}, vk_{\mathcal{A}^{l}_{sum}} \leftarrow \mathsf{Setup}(1^\lambda, \mathcal{A}^{l}_{sum}, \theta_{\mathcal{A}^{l}_{sum}})$

        \LineComment{salted classification assessment \cosmetic on test sample}
        \State $\Psi, H^{root} \leftarrow$ \Call{CosmeticCROBuild}{$U, U^{test}, vk_{\mathcal{L}^{s}_{CA}}, vk_{\mathcal{A}^{l}_{sum}}$}

        \State $pk_{ACC},vk_{ACC} \leftarrow \mathsf{Setup}(1^\lambda, \mathcal{G}_{ACC})$
        \State distribute $vk_{ACC}$ to all users in $U^{test}$
        \State set $\Omega_{ACC} = \{\Psi\}$
        \State generate $\Pi_{ACC} = \mathsf{Prove}(pk_{ACC},\Omega_{ACC})$
        \State distribute $\Pi_{ACC}, pk_{ACC}$ publicly
    \EndFunction
 \end{algorithmic}
\end{algorithm}
\subsection{Logistic Accuracy (ACC) Test}\label{subsec:acc_test}
Finally, we evaluate logistic regression accuracy (ACC) under privacy-preserving computation (Algorithm~\ref{alg:acc_algorithm}). The classification assessment leaf transformation $\mathcal{L}^{s}_{CA}$, parameterized by the logistic regression coefficients $\theta_{\mathcal{L}^{s}_{CA}}$, returns 1 if the predicted class matches the true label and 0 otherwise; the aggregation function $\mathcal{A}^{l}_{sum}$ sums these values and divides by sample size to yield overall accuracy. The LTR verifier confirms correctness by checking the CSMT-based inclusion/exclusion proof for $u$ and the classifier-aggregation proof $\Pi_{CA}$ against $vk_{CA}$, returning success only if both pass and otherwise failure tagged with the failing check(s).
        
\section{Supplementary Results}\label{apx:sup_res}
\textit{Cryptographic Overhead and Key Size}: Proving and verification key sizes are constant across EZKL scales 8--14 for the KS, LR, and Acc tests, showing no scale-dependent cryptographic growth. LTR-class transformers converge to 11.27~GB/2.69~MB and MRP-class transformers to 10.29~GB/2.43~MB across all three tests; the sole exception is the Acc transformer, which rises modestly to 13.24~GB/3.07~MB due to its additional comparative logic.

\end{document}